\documentclass[a4paper,11pt]{article}
\usepackage[utf8]{inputenc}

\usepackage{amsmath}
\usepackage{amsthm}
\usepackage{amssymb}

\usepackage{graphicx}
\usepackage{xcolor}

\newtheorem{lemma}{Lemma}
\newtheorem{theorem}{Theorem}
\newtheorem{corollary}{Corollary}
\newtheorem{definition}{Definition}
 
\allowdisplaybreaks

\title{Constructing Long Paths in Graph Streams}
\author{
Christian Konrad\footnote{Supported by EPSRC New Investigator Award EP/V010611/1.} \\ 
University of Bristol \\ 
\texttt{christian.konrad@bristol.ac.uk} 
\and
Chhaya Trehan\footnote{Supported by EPSRC New Investigator Award EP/V010611/1.} \\ 
 Unaffiliated Researcher\\ 
\texttt{chhaya.dhingra@gmail.com} 
}

\date{}
\definecolor{darkgreen}{rgb}{0, 0.5, 0}

\DeclareMathOperator{\poly}{poly}
\DeclareMathOperator{\ICost}{ICost}
\DeclareMathOperator{\IC}{IC}
\DeclareMathOperator{\lp}{lp}
\DeclareMathOperator{\lc}{lc}
\DeclareMathOperator{\Exp}{{\mathop{\mathbb{E}}}}% expected value

\usepackage{algorithm}
\usepackage{algorithmic}
\usepackage{fullpage}

\usepackage[margin=1in]{geometry}

\begin{document}

\maketitle
\thispagestyle{empty}

\begin{abstract}

In the {\em graph stream model of computation}, an algorithm processes the edges of an $n$-vertex input graph in one or more sequential passes while using a memory that is sublinear in the input size. 

\smallskip

The streaming model poses significant challenges for algorithmically constructing long paths. Many known algorithms that are tasked with extending an existing path as a subroutine require an entire pass over the input to add a single additional edge. 
 This raises a fundamental question: Are multiple passes inherently necessary to construct paths of non-trivial lengths, or can a single pass suffice? To address this question, we systematically study the \textsf{Longest Path} problem in the one-pass streaming model. In this problem, given a desired approximation factor $\alpha$, the objective is to compute a path of length at least $\lp(G) / \alpha$, where $\lp(G)$ is the length of a longest path in the input graph $G$.

\smallskip

We give algorithms as well as space lower bound results for both undirected and directed graphs. Besides the insertion-only model, where the input stream solely consists of the edges of the input graph, we also study the insertion-deletion model, where previously inserted edges may be deleted again. Our results include:

\begin{enumerate}
 \item We show that for undirected graphs, in both the insertion-only and the insertion-deletion streaming models, there are {\em semi-streaming algorithms}, i.e., algorithms that use space $O(n \poly \log n)$, that compute a path of length at least $d /3$ with high probability, where $d$ is the average degree of the input graph. These algorithms can also yield an $\alpha$-approximation to \textsf{Longest Path} using space $\tilde{O}(n^2 / \alpha)$.

 \item Next, we show that such a result cannot be achieved for directed graphs, even in the insertion-only model. We show that computing a $(n^{1 - o(1)})$-approximation to \textsf{Longest Path} in directed graphs in the insertion-only model requires space $\Omega(n^2)$. This result is in line with recent results that demonstrate that processing directed graphs is often significantly harder than undirected graphs in the streaming model.
 
 \item We further complement our results with two additional lower bounds. First, we show that semi-streaming space is insufficient for small constant factor approximations to \textsf{Longest Path} for undirected graphs  in the insertion-only model. 
 Last, in undirected graphs in the insertion-deletion model, we show that computing an $\alpha$-approximation requires space $\Omega(n^2 / \alpha^3)$.
\end{enumerate}
 \end{abstract}
 \clearpage
 \pagenumbering{arabic} 
 \newpage

 \section{Introduction}
  In the {\em graph stream model of computation}, an algorithm processes a stream of edge insertions and deletions that make up an $n$-vertex input graph $G=(V, E)$ via one or multiple sequential passes. The primary objective is to design algorithms that use as little space as possible. 
 
 The most studied and best-understood setting is the one-pass insertion-only setting, where only a single pass is allowed and the input stream does not  contain any edge deletions. It is known that many fundamental graph problems can be solved well in this setting, e.g., there are {\em semi-streaming algorithms}, i.e., algorithms that use space $\tilde{O}(n) = O(n \poly \log n)$\footnote{We write $\tilde{O}(.)$ to mean the usual Big-$O$ notation with poly-logarithmic dependencies suppressed. $\tilde{\Theta}(.)$ and $\tilde{\Omega}(.)$ are defined similarly.}, for computing a spanning tree and a maximal matching \cite{fkmsz05}, while no $o(n^2)$ space algorithms exist for other problems, such as computing a BFS or DFS tree and computing a maximal independent set \cite{cdk19,ack19}. % and determining the distance between two given vertices. 
 For such problems, algorithms that perform multiple passes over the input are then usually considered \cite{km19,cfht20,acgmw21,akns24}.
 
 Our work is motivated by the observation that streaming algorithms that construct long paths, often as a subroutine, require a large number of passes. For example, many streaming algorithms for computing large matchings construct augmenting paths as a subroutine (e.g., \cite{m05,kns23,mmss25}), and these algorithms typically only add a single edge or a few edges per pass to any not-yet-completed augmenting path. Another example are streaming algorithms for computing BFS/DFS trees that extend partial trees/paths by adding only a single edge per pass \cite{km19,cfht20}. This raises a fundamental question: Is the one/few-edges-per-pass strategy best possible and multiple passes are inherently necessary to construct paths of non-trivial lengths, or can a single pass suffice?

 \vspace{-0.3cm}
 \subparagraph{The \textsf{Longest Path} Problem}
 We address this question by systematically studying the \textsf{Longest Path} (\textsf{LP}) problem in the one-pass streaming setting. In this problem, the objective is to compute a path of length $\lp(G) / \alpha$, where $\lp(G)$ denotes
 the length of a longest path in the input graph $G$, and $\alpha \ge 1$ is the approximation factor. 
 
 In the offline setting, in undirected graphs, it is known that the problem is NP-hard and hard to approximate in the sense that outputting a path of length $n-n^{\epsilon}$ in Hamiltonian graphs is NP-hard \cite{kmr97}, for any $\epsilon  < 1$. In directed graphs, a much stronger hardness result is known. Bj\"{o}rklund, Husfeldt and Khanna \cite{bhk04} showed that it is NP-hard to approximate \textsf{LP} in directed graphs within a factor of $n^{1-\epsilon}$, for any $\epsilon > 0$. Regarding upper bounds, 
it is known that a path of length $d_{\text{min}}$ can be constructed greedily in polynomial time, where $d_{\text{min}}$ is the min-degree of the input graph \cite{kmr97}. There are also FPT algorithms with runtime polynomial in $n$ but exponential in the length of the path constructed (see, for example, \cite{bhkk17} and the references therein).  
 We note that streaming algorithms for the \textsf{LP} problem have also previously received  attention, however, only from a practical perspective. Kliemann et al. \cite{kss16} studied practical multi-pass algorithms without providing any theoretical guarantees. 
 
 \subsection{Our Results}
 In this work, we give one-pass streaming algorithms and space lower bounds for undirected graphs and a strong space lower bound for directed graphs. We consider both the insertion-only model (no deletions) and the insertion-deletion model (deletions allowed). 
 
 As our first result, we show that, in undirected graphs, if we sample $O(n \log n)$ random edges from the input graph, then this sample contains a path of length $\Omega(d)$ with high probability, where $d$ is the average degree of the input graph. Since this sampling task can be implemented in both the insertion-only and the insertion-deletion streaming models using standard techniques, we obtain the following theorem:
 
 \begin{theorem} \label{thm:alg-1}
 In both the insertion-only and the insertion-deletion models, for undirected graphs, there are semi-streaming (i.e. $O(n \poly \log n)$ space) algorithms that compute a path of length at least $d/3$ with high probability, where $d$ is the average degree of the input graph.
\end{theorem}
\newcounter{counterALG}  
\setcounter{counterALG}{\value{theorem}}

 We remark that our algorithm can be used to obtain an $\alpha$-approximation algorithm to \textsf{LP} using $\tilde{O}(n^2 / \alpha)$ space, for any $\alpha \ge 1$. This is achieved by running our algorithm in parallel with the trivial algorithm that stores $\tilde{\Theta}(n^2 / \alpha)$ edges. Then, if the space constraint of $\tilde{\Theta}(n^2 / \alpha)$ is large enough to store the entire graph then we can find an exact solution. If not, then we are guaranteed that the average degree of the input graph is $\Omega(n / \alpha)$, which implies that the algorithm of Theorem~\ref{thm:alg-1} yields the desired result. We thus obtain the following corollary:

\newcounter{counterCOR}  
\setcounter{counterCOR}{\value{theorem}}

 \begin{corollary}\label{cor:upper-bound}
 In both the insertion-only and the insertion-deletion models, for undirected graphs, there are $\tilde{O}(n^2 / \alpha)$-space streaming algorithms  that compute an $\alpha$-approximation to \textsf{Longest Path}.
\end{corollary}

 Next, we ask whether a similar algorithmic result is possible in directed graphs. As our main and most technical result, we show that this is not the case in a strong sense:
 
 \begin{theorem}\label{thm:lb-directed}
  Every one-pass streaming algorithm for \textsf{Longest Path} in directed graphs with approximation factor $n^{1-o(1)}$ requires space $\Omega(n^2)$.
 \end{theorem}
 \newcounter{counterLB-directed}  
\setcounter{counterLB-directed}{\value{theorem}}

Theorem~\ref{thm:lb-directed} together with Corollary~\ref{cor:upper-bound} establish a separation in the space complexity between algorithms for undirected and directed graphs for \textsf{LP} in the insertion-only model. This lower bound is also in line with recent results on streaming algorithms for directed graphs that demonstrate that problems on general directed graphs are often hard to solve in the streaming model \cite{cgmv20}. 

Finally, we complement our results with two additional lower bound results. First, we give a lower bound for insertion-only streams and undirected graphs, ruling out the existence of semi-streaming algorithms with constant approximation factor close to $1$.

\begin{theorem}
 Every one-pass insertion-only streaming algorithm for \textsf{Longest Path} on undirected graphs with approximation factor $1+\frac{1}{25} - \gamma$, for any $\gamma > 0$, requires space $n^{1 + \Omega(\frac{1}{\log \log n})}$. 
\end{theorem}
\newcounter{counterLBundirected}  
\setcounter{counterLBundirected}{\value{theorem}}

We note that this lower bound result is significantly weaker than our lower bound for directed graphs, both in terms of approximation factor and space. This, however, is in line with the status of the problem in the offline setting, where a very strong impossibility result for directed graphs is known, but only significantly weaker impossibility results for undirected graphs exist.  

Last, we show that, in insertion-deletion streams, space $\Omega(n^2 / \alpha^3)$ is required for computing an $\alpha$-approximation.

\begin{theorem}
 Every one-pass insertion-deletion streaming algorithm for \textsf{Longest Path} on undirected graphs with approximation factor $\alpha \ge 1$ requires space $\Omega(n^2 / \alpha^3)$. 
\end{theorem}
\newcounter{counterLB-deletion}  
\setcounter{counterLB-deletion}{\value{theorem}}

 This lower bound together with our algorithm show that the optimal dependency of the space complexity on $\alpha$ in insertion-deletion streams is between $1/ \alpha$ and $1 / \alpha^3$.

 \subsection{Our Techniques}

 \subparagraph{Algorithm.} We will first explain the key ideas behind our algorithm. 
 As observed by Karger et al. \cite{kmr97}, a path of length $d_{\text{min}}$, where $d_{\text{min}}$ is the minimum degree of the input graph, can be computed as follows: Start at any vertex $v_0$ and visit an arbitrary neighbor that we denote by $v_1$. In a general step $i$, we have already constructed the path $v_0, \dots, v_i$. We then visit any neighbor $v_{i+1}$ of $v_i$ that has not previously been visited. Then, as long as $i < d_{\text{min}}$, we can always find a yet unvisited neighbor, and, hence, we obtain a path of length  $d_{\text{min}}$.

 We first see that this argument can also be applied to the average degree $d$ of the input graph $G$. It is well known that, by repeatedly removing vertices of degree at most $d/2$ from $G$ until no such vertex remains, we are left with a non-empty graph $G'$ with min-degree at least $d/2$. We can now apply the same argument as above to $G'$ and thus find a path of length at least $d/2$. 

 The approach outlined above, however, does unfortunately not yield a small space streaming algorithm since a) we do not know which vertices are contained in $G'$, and b) we cannot afford to store $\Theta(d)$ incident edges on each vertex.

 To overcome these obstacles, we resort to randomization. Our algorithm solely samples $O(n \log n)$ random edges $F$ from the input graph and outputs a longest path among the edges $F$. To see that a long path in $F$ exists, we argue that the subset of edges $F' \subseteq F$ that are also contained in $G'$ contains a path of length at least $d/3$ with high probability. Such a path can be constructed greedily. Suppose we have already constructed a partial path of length $\ell < d/3$ solely using the edges $F'$, and let $v_{\ell+1} \in V(G')$ denote its current endpoint. Then, since $v_{\ell+1}$ has a degree of at least $d/2$ in $G'$, there are at least $d/2 - d/3 = d/6$ neighbors of  $v_{\ell+1}$ in $G'$ that have not yet been visited in the path. Since we sample $\Theta(n \log n)$ edges overall, the probability that any one of these edges incident to $v_{\ell+1}$ that connect to these $d/6$ vertices is sampled is $\Omega(\frac{n \cdot \log n}{n \cdot d}) =  \Omega(\frac{\log n}{d})$, which implies that at least one of these edges is sampled with high probability and we can extend the path. 

 \vspace{-0.3cm}
\subparagraph{Space Lower Bounds.} All our space lower bounds are proved in the one-way two-party communication setting. In this setting, two parties that we denote by Alice and Bob each hold a subset of the edges of the input graph $G = (V, E = E_A \cup E_B)$, with $E_A$ being Alice's edges, and $E_B$ being Bob's edges. Alice sends a single messages $\Pi$ to Bob, and Bob computes the output of the protocol. Then, it is well-known that a lower bound on the size of the message $\Pi$ also constitutes a lower bound on the space of any one-pass streaming algorithm.

We work with {\em induced matchings} in all our lower bound constructions. In a graph $G=(V, E)$, a matching $M \subseteq E$ is {\em induced} if the edges of the vertex-induced subgraph $G[V(M)]$ are precisely the edges $M$. Suppose now that the input graph is bipartite, and we denote it by $G=(A, B, E=E_A \cup E_B)$. Furthermore, we suppose that Alice's subgraph, i.e., the graph spanned by the edges $E_A$, contains a matching $M$ that is induced in the final graph. We say that $M$ is the {\em special} matching. %Furthermore, suppose that Bob holds an additional random  matching $N$ in the subgraph $G[V(M)]$ that matches all vertices $V(M)$. Then, it can be seen that in expectation over the choice of $N$ the set $M \cup N$ contains a path of length $\Omega(|M|)$. 

Our goal is to complete our lower bound constructions so that:
\begin{enumerate}
    \item Every long path in the graph must contain many edges of the special matching $M$;
    \item The special matching $M$ is {\em hidden} among the edge set $E_A$ so that Alice cannot identify $M$, and, given a limited communication budget, Alice therefore cannot forward many of $M$'s edges to Bob.
\end{enumerate}
The two properties then imply a lower bound since, if Bob  knows only few edges of $M$, but every long path contains many such edges, then Bob cannot output a long path.

To achieve property $2$, in our insertion-deletion lower bound, we make use of edge deletions as part of Bob's input to turn a large matching in $E_A$ into an induced matching, and in our insertion-only lower bounds, we work with {\em Rusza-Szemer\'{e}di graphs} (RS-graphs in short), which have been extensively used for proving lower bounds for matching problems in the streaming setting (e.g., \cite{gkk12,akl16,kn21,kn24}). An $(r,t)$-RS-graph is a balanced bipartite graph on $2n$ vertices such that its edge set can be partitioned into $t$ induced matchings $M_1, \dots, M_t$, each of size $r$. Our insertion-only constructions are such that each of these $t$ matchings can take the role of the special induced matching $M$, which will allow us to argue that Alice essentially has to send many edges of each induced matching to Bob if Bob is able to report a long path.

Regarding establishing property $1$, we pursue different strategies that depend on the specific streaming setting, and on whether we work with directed or undirected graphs. These strategies are described below and illustrated in Figure~\ref{fig:techniques}.

\begin{figure}[h]
\fbox{\begin{minipage}{\textwidth}
\vspace{0.1cm}
\textbf{Directed Graphs} \hfill  \textbf{Undirected Graphs} \hfill \textbf{Insertion-deletion}

\vspace{0.3cm}
\includegraphics[width=\textwidth]{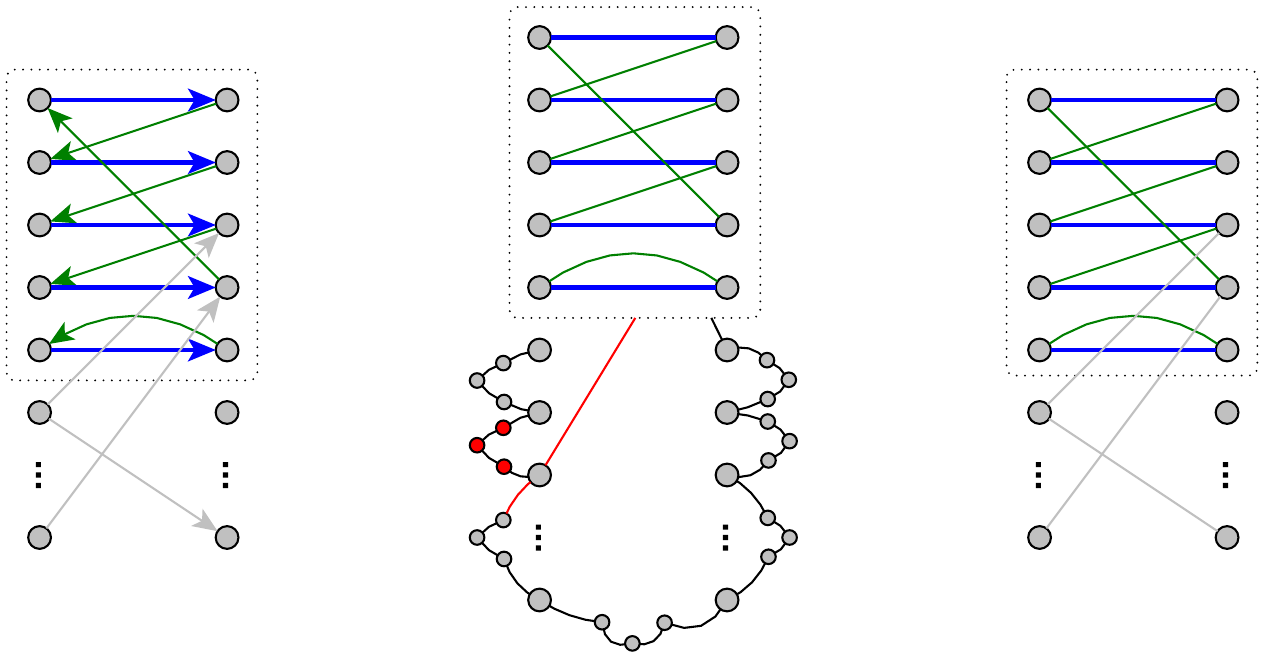}

\vspace{-6cm} \hspace{3.2cm} $\color{blue} M_J \color{black} \cup \color{darkgreen}N$ \hspace{6.2cm} 
$\color{blue} M \color{black} \cup \color{darkgreen}N$ 

\vspace{-1.3cm} \hspace{9.4cm} $\color{blue} M_J \color{black} \cup \color{darkgreen}N$

\vspace{3.7cm}

\hspace{6.9cm} $v$

\vspace{-1cm}
\hspace{11.2cm} $O(\frac{n}{\alpha})$ \scalebox{1.5}[1.5]{$\Biggl\{$}

\vspace{0.2cm}

\hspace{0.17cm} $A$ \hspace{1.8cm} $B$

\vspace{-0.1cm} \hspace{9.4cm} $\mathcal{P}$

\vspace{0.6cm}
\vspace{0.1cm} 

    \textbf{Directed Graphs:} Since $N$ constitutes the only edges directed from right to left, and $M_J$ is an induced matching, up to potentially the first and last edge, every long path alternates between edges of $M_J$ and $N$ and thus contains many edges of $M_J$. \vspace{0.1cm}\\
    \textbf{Undirected Graphs:} (Insertion-Only) The path $\mathcal{P}$ is connected to a longest path in $V(M_J)$. Our argument is based on the fact that it is detrimental to the length of the output path to include many edges between $V \setminus V(M_J)$ and $V(M_J)$. Indeed, suppose that the output path contains the two red edges incident on vertex $v$. Then, the vertices highlighted in red  cannot be visited by the output path. This allows us to argue that, for each edge going across the cut $(V \setminus V(M_J), V(M_J))$, multiple vertices on the path $\mathcal{P}$ are not visited. This in turn implies that every long path must contain many edges of $M_J$ in order to visit the vertices in $V(M_J)$. \vspace{0.1cm}\\
    \textbf{Insertion-Deletion Streams:} (Undirected Graphs) Our construction is such that there are only $O(\frac{n}{\alpha})$ vertices outside $V(M)$. Hence, if Bob does not know many edges of $M$ then there are at most $O(\frac{n}{\alpha})$ vertices of $V \setminus V(M)$ available to connect to vertices within $V(M)$. Since this number is small, Bob is required to learn many edges of $M$ in order to output a long path.  \vspace{0.1cm}
    \end{minipage}}
    \caption{Illustrations of our three lower bound constructions.  \label{fig:techniques}}

\end{figure}

\textit{Directed Graphs in the Insertion-only Model.}
The cleanest case is our lower bound for directed graphs. In this setting, Alice holds the edges of an RS-graph $G'=(A, B, E)$, and we assume that all edges are  directed from $A$ to $B$. We then pick a matching $M_J$, for a uniform random $J \in [t]$, and Bob inserts a random matching $N$ that matches $B(M_J)$ to $A(M_J)$ so that all edges in $N$ are directed from $B$ to $A$. We then leverage the well-known result that the longest cycle in a random permutation of the set $[r]$ is of length $\Theta(r)$ \cite{gwg59,g64} in order to argue that $M_J \cup N$ contains a path of expected length $\Theta(|M_J|)$. It then remains to argue that every path of length $\ell$, for any $\ell$, in the input graph $G = G' \cup N$ must contain $\Omega(\ell)$ edges of $M_J$, which uses the fact that the matching $M_J$ is induced. Since, however, Alice did not know that $M_J$ is the special matching, Alice is required to send a large number of edges of each induced matching to Bob so that Bob can output a long path. 

For technical reasons, our graph construction is slightly more involved than described above since we need to turn Alice's RS-graph into an input distribution, see Section~\ref{sec:lb-directed} for details. While it is easy to argue that if Alice sends $k$ edges of $M_J$ to Bob then Bob can compute a path of length at most $O(k)$ in $M_J \cup N$, we need to argue a much stronger bound. We show that, even if Alice sends as many as $r/100$ $M_J$ edges to Bob then Bob can still only find a path of size $O(\log r)$ in $M_J \cup N$. 
We implement our approach using the information complexity paradigm, including  direct sum and message compression arguments. 

\textit{Undirected Graphs in the Insertion-only Model.} In our construction for undirected graphs, Alice also holds an RS-graph $G'=(A, B, E)$, and, given a random index $J \in [t]$, Bob also holds a random matching $N$ matching $A(M_J)$ to $B(M_J)$. Similar as above, the matchings $M_J \cup N$ contain a path of length $\Omega(|M_J|)$. However, it is no longer true that any long path in $G' \cup N$ must contain many edges of $M_J$ as, for example, the vertices in $V \setminus V(M_J)$ can now be used to visit the vertices within $V(M_J)$ as their incident edges are undirected.

Our aim is to ensure that it is detrimental for constructing long paths to include many edges across the cut $(V \setminus V(M_J), V(M_J))$ in the path. Once this property is established, we will argue that, if Bob only knows few of the $M_J$ edges then many vertices of $V(M_J)$ will remain unvisited in the output path, which implies that the path is bounded in length as not all vertices are visited. To argue that only few edges across the cut $(V \setminus V(M_J), V(M_J))$ are included in the output path, we ensure that all the vertices $v \in V \setminus V(M_J)$ serve as {\em gateways} to other parts of the graph that are introduced by Bob. The construction is so that these other parts cannot be visited if $v$ connects to a vertex in $V(M_J)$ in the output path, which will then be detrimental for the construction of  a long path. To achieve this, Bob introduces additional edges as follows. First, let $\mathcal{P}'$ be a path consisting of novel edges that visits every vertex in $V \setminus V(M_J)$, and let $\mathcal{P}$ be the path obtained from $\mathcal{P}'$ by subdividing every edge in $\mathcal{P}'$, $\ell$ times, for some integer $\ell$. Furthermore, this path is connected to the longest path within $M_J \cup N$. Then, the path $\mathcal{P}$ significantly contributes to the longest path in the input graph. Observe, however, if a vertex $v \in V \setminus V(M_J)$ connects to a vertex in $V(M_J)$ then the vertices on the subdivision on one of the edges of $\mathcal{P}$ incident on $v$ cannot be visited anymore. By setting $\ell$ large enough, we show that it is detrimental for the algorithm to use such vertices to connect to $V(M_J)$, which renders the edges of $M_J$ indispensable for a long path. 

As above, the actual lower bound construction is more involved since we need to turn Alice's RS-graph into an input distribution. On a technical level, we give a reduction to the two-party communication problem \textsf{Index} using ideas similar to those introduced by Dark and Konrad \cite{dk20}. For details, see Section~\ref{sec:lb-undirected}.

\textit{Undirected Graphs in the Insertion-deletion Model.} The key advantage that allows us to prove a much stronger lower bound for insertion-deletion streams than for insertion-only streams is that Bob can insert edge deletions that turn a large matching contained in Alice's edges $E_A$ into an  induced matching. We see that it is possible to achieve this so that: 
\begin{enumerate}
    \item The special (induced) matching $M$ is of size $n - O(n/\alpha)$; and
    \item The special matching $M$ is hidden among Alice's edges.
\end{enumerate}
Furthermore, we make sure that there exists a second matching $N$ such that $N \cup M$ forms a path of length $\Omega(n)$.

Property~1 significantly helps in proving our lower bound since, as opposed to our lower bound for undirected graphs in insertion-only streams, there are only $O(n/\alpha)$ vertices outside the set $V(M)$, and, hence, these vertices only allow us to visit $O(n / \alpha)$ vertices of $V(M)$. We can then argue that, for the remaining $n - O(n / \alpha)$ vertices of $V(M)$, Bob is required to know the edges of $M$ for those to be visited. This however requires Alice to send these edges to Bob. Then, Property~2 ensures that Alice cannot identify these edges and would therefore have to send most edges of the input graph to Bob, which exceeds the allowed communication budget.

On a technical level, we give a reduction to the \textsf{Augmented-Index} two-party communication problem, again, reusing many of the ideas given in the lower bound by Dark and Konrad \cite{dk20}. For details, please refer to Section~\ref{sec:lb-insert-delete}.

\subsection{Outline}
We provide notation, state important RS-graph constructions, discuss the  \textsf{Index} and \textsf{Augmented-Index} communication problems, introduce the information complexity framework, state important inequalities involving mutual information, and give a message compression theorem in our preliminaries section, Section~\ref{sec:prelim}. Then, we present our algorithm in Section~\ref{sec:ub} and our lower bound for directed graphs in insertion-only streams is given in Section~\ref{sec:lb-directed}. Our lower bound for undirected graphs in insertion-only streams is presented in Section~\ref{sec:lb-undirected} and our lower bound for undirected graphs in insertion-deletion streams 
is presented in Section~\ref{sec:lb-insert-delete}.
We conclude in Section~\ref{sec:conclusion} with open problems. 
 
\section{Preliminaries}\label{sec:prelim}
%\subsection{The Longest Path Problem}
%Given a undirected or directed graph $G=(V, E)$, we denote by $\lp(G)$ the length of a longest path in $G$. We study the \textsf{Longest Path} (\textsf{LP}) problem, which consists of computing a path of length $\lp(G)$. We say that an algorithm for \textsf{LP} has an approximation factor of $\alpha \ge 1$ if it outputs a longest path of length at least $\lp(G) / \alpha$.

Given a graph $G=(V, E)$, for vertices $u,v \in V$, we denote an undirected edge between $u$ and $v$ by $\{u,v\}$, and a directed edge with tail $u$ and head $v$ by $(u,v)$.
%\vspace{-0.2cm}
\subsection{Rusza-Szemer\'{e}di Graphs}
In our lower bounds, we make use of Rusza-Szemer\'{e}di graphs.

\begin{definition}[Rusza-Szemer\'{e}di Graph]\label{def:rs-graphs}
A bipartite graph $G = (A, B, E)$ is 
an {\em $(r, t)$-Ruzsa-Szemer\'{e}di graph}, $(r, t)$-RS graph for short, if its edge set can be partitioned into $t$ induced matchings, each of size $r$. 
\end{definition}

We will use the RS-graph constructions of Alon et al. \cite{ams12} and Goel et al. \cite{gkk12}, the latter is based on the construction by Fischer et al. \cite{flnrrs02}.

\begin{theorem}[Ruzsa-Szemer\'{e}di Graph Constructions]\label{thm:rs-graphs}
 There are bipartite RS-graphs $G=(A, B, E)$ with $|A| = |B| = n$ %where the edge set $E = \ M_1 \ \dot{\cup} \ M_2 \ \dot{\cup} \ \dots \ \dot{\cup} \ M_t$ is a collection of disjoint induced matchings $M_i$, each of size $|M_i| = r$, for all $i$, 
 with the following parameters:
 \begin{enumerate}
  \item $r = n^{1-o(1)}$, and $t \cdot r = \Omega(n^2)$ \cite{ams12}; and 
  \item $r = (\frac{1}{2}-\epsilon) n$, for any constant $\epsilon > 0$, and $t = n^{\Omega(\frac{1}{\log \log n})}$ \cite{gkk12}. 
 \end{enumerate}
\end{theorem}

%\vspace{-0.2cm}
\subsection{Streaming Models}
Given an input graph $G = (V, E)$, an {\em insertion-only} stream describing $G$ is an arbitrarily ordered sequence of the edges $E$. Then, an {\em insertion-deletion} stream is a sequence of edge insertions and deletions so that, at the end of the stream, the surviving edges constitute the edge set $E$. Furthermore, the stream is such that only edges that have previously been inserted are deleted, i.e., the multiplicity of every edge is never negative. We also assume that, at any moment, the multiplicity of any edge is polynomially bounded. The latter property is a standard assumption and required so that sampling methods, such as $\ell_0$-sampling, require only $\poly \log(n)$ space.  

%\vspace{-0.2cm}
\subsection{Communication Complexity}
We consider the one-way two-party model of communication for proving our space lower bounds. In this setting, two parties, denoted Alice and Bob, share the input data $X = (X_A, X_B)$ so that Alice holds $X_A$ and Bob holds $X_B$.
They operate as specified in a protocol $\Pi$ in order to solve a problem $\mathcal{P}$. Alice and Bob can make use of both private and public randomness. Randomness is provided via infinite sequences of uniform random bits. Private randomness can only be accessed by one party. The sequence of public randomness can be accessed by both Alice and Bob, and we denote this sequence by $R$.

The protocol $\Pi$ instructs the parties to operate as follows. First, Alice computes a message that we (ambiguously) also denote by $\Pi$ as a function of $X_A$, $R$, and her private randomness, sends the message $\Pi$ to Bob, who then computes the output of the protocol as a function of $X_B, \Pi, R$ and his private randomness.

The {\em communication cost} of a protocol $\Pi$ is the maximum length of the message sent from Alice to Bob in any execution of the protocol. Then, the {\em communication complexity} of a problem $\mathcal{P}$ is the minimum communication cost of any protocol that solves $\mathcal{P}$. 

%\vspace{-0.2cm}
\subsection{The \textsf{Index} and \textsf{Augmented-Index} Problems}
In this work, we will exploit known lower bounds on the well-known  two-party communication problems \textsf{Index} and \textsf{Augmented-Index}.

\begin{definition}[\textsf{Index}$(N)$ and \textsf{Augmented-Index}$(N)$]\label{def:index}
 In the one-way two-party communication problem \textsf{Index}$(N)$, Alice holds a bit-vector $X \in \{0, 1\}^N$ of length $N$ and Bob holds an index $J \in [N]$. Alice sends a message to Bob, who, upon receipt, outputs the bit $X[J]$. 
 In \textsf{Augmented-Index}$(N)$, the setup is identical, however, Bob additionally knows the suffix $X[J+1, N]$. 
\end{definition}
It is well-known that both problems require a message of size $\Omega(N)$, even if the error probability is as large as $\frac{1}{2} - \delta$, for any constant $\delta > 0$.

 \begin{theorem}[Lower Bounds for \textsf{Index} and \textsf{Augmented-Index}]\label{thm:lb-index}
 Every randomized protocol for \textsf{Index}$(N)$ or \textsf{Augmented-Index}$(N)$ that succeeds with probability at least $\frac{1}{2} + \delta$, for any $\delta > 0$, has a communication cost of $\Omega(N)$.
\end{theorem}
A proof of this theorem is given in \cite{ry20}. This proof is stated for \textsf{Index} but also applies to \textsf{Augmented-Index} with minor modifications.

%\vspace{-0.2cm}
\subsection{Information Complexity and Message Compression}
We prove lower bounds using the {\em information complexity paradigm}, which is a framework that is based on information theory. Let $(A, B, C) \sim \mathcal{D}$ be jointly distributed random variables according to distribution $\mathcal{D}$. We denote the {\em Shannon entropy} of $A$ by $H_{\mathcal{D}}(A)$, the entropy of $A$ conditioned on $B$ by $H_{\mathcal{D}}(A \ | \ B)$, the {\em mutual information} of $A$ and $B$ by $I_{\mathcal{D}}(A \ : \ B)$, and the conditional mutual information between $A$ and $B$ conditioned on $C$ by $I_{\mathcal{D}}(A \ : \ B \ | \ C)$. We may also drop the subscript $\mathcal{D}$ in $H_{\mathcal{D}}(.)$ and $I_{\mathcal{D}}(.)$ if it is clear from the context.

We will use the following standard facts about entropy and mutual information: 
(let $(A,B,C,D) \sim \mathcal{D}$ be jointly distributed random variables.)
\begin{enumerate}
 \item[\textbf{P1:}] If $A$ and $C$ are independent conditioned on $D$ then: $I(A \ : \ B \ | \ D) \le I(A \ : \ B \ | \ C, D)  $ 
 \item[\textbf{P2:}] $I(A \ : \ B \ | \ C,D) = \Exp_{d \gets D}I(A \ : \ B \ | \ C,D = d)$
 \item[\textbf{P3:}] Let $E$ be an event independent of $A, B, C$. Then: $I(A \ : B \ | \ C, E) = I(A \ : B \ | \ C)$
 \item[\textbf{P4:}] $I(A,B \ : C \ | \ D) = I(A \ : C \ | \ D) + I(B \ : C \ | \ D,A)$
\end{enumerate}

Given a one-way two-party communication protocol $\Pi$ and an input distribution $(X_A, X_B) \sim \mathcal{D}$, we will measure the amount of information that the message $\Pi$ reveals about Alice's input under distribution $\mathcal{D}$. The following quantity is denoted the {\em external information cost} of $\Pi$:

\begin{definition}
 The (external) information cost $\ICost_{\mathcal{D}}(\Pi)$ of the one-way two-party communication protocol $\Pi$ under input distribution $\mathcal{D}$ is defined as:
 $$\ICost_{\mathcal{D}}(\Pi) = I_{\mathcal{D}}(X_A \ : \ \Pi \ | \ R) \ . $$ 
\end{definition}

Then, for a given problem $\mathcal{P}$, we denote the {\em information complexity} $\IC_{\mathcal{D}}(\mathcal{P})$ of $\mathcal{P}$ under distribution $\mathcal{D}$ as the minimum information cost of a protocol that solves $\mathcal{P}$.

It is well-known that information cost of a protocol is a lower bound on communication cost of the protocol. 

We will also use a {\em message compression} result, which is due to Harsha et al. \cite{hjmr07}. We follow the presentation of this result given in \cite{ry20}.

\begin{theorem}[Message Compression]\label{thm:compression}
 Let $\Pi$ be a protocol in the one-way two-party communication setting. Then, the protocol can be simulated with a different protocol that sends a message of expected size at most
 $$I_{\mathcal{D}}(X_A \ : \ \Pi) + 2 \cdot \log \left(1 + I_{\mathcal{D}}(X_A \ : \ \Pi)  \right) + O(1) \ . $$
\end{theorem}

\section{Algorithm}\label{sec:ub}
We will first describe and analyze our sampling-based algorithm in Subsection~\ref{sec:sampling-based}, and then discuss implementations of this algorithm in the streaming models in Subsection~\ref{sec:implementations}.
\subsection{Sampling-based Algorithm} \label{sec:sampling-based}
Let $G = (V, E)$ be the input graph with $n = |V|$, $m = |E|$, and average degree $d = 2 \frac{m}{n}$.

Consider the following algorithm:

\begin{algorithm}
 \begin{algorithmic}
  \REQUIRE $G = (V, E)$
  \STATE Sample $10 \cdot n \cdot \ln n$ random edges from $E$ and denote this set by $F$
  \RETURN Longest path in $G[F]$ 
 \end{algorithmic}
 \caption{Sampling-based algorithm for constructing a path of length $\Omega(d)$. \label{alg:sampling}}
\end{algorithm}

We show next that Algorithm~\ref{alg:sampling} constructs a path of length at least $d/3$ with high probability.

\begin{theorem}
 Algorithm~\ref{alg:sampling} constructs a path of length $d/3$ with probability at least $1-\frac{1}{n^2}$. It can also be regarded as an $O(\frac{n}{d})$-approximation algorithm for \textsf{Longest Path}.
\end{theorem}

\begin{proof}
Given the input graph $G = (V, E)$, let $U \subseteq V$ denote a subset of vertices of $V$ such that $G[U]$ has minimum degree at least $d / 2$. It is well-known that such a subset of vertices exists, and we give a proof of this statement for completeness in the appendix (Lemma~\ref{lem:avg-min}).

Let $u_0 \in U$ be any vertex, and let $P_0 = \{u_0 \}$ be the path of length $0$ with start and end point  $u_0$. We will extend $P_0$ using the edges in $F$ that are also contained in $G[U]$, as follows:

In step $i = 1, 2, \dots, d / 3$, we add the edge $(u_{i-1}, u_i)$ to $P_{i-1}$ and obtain the path $P_i$. The edge $(u_{i-1}, u_i)$ is an arbitrary edge in $F$ incident on $u_{i-1}$ that connects to a vertex $u_i \in U$ that has not yet been visited on the path $P_{i-1}$. We will now prove that such a vertex exists.

First, recall that, by definition of $U$, the vertex $u_{i-1}$ has at least $d / 2$ neighbors in $G[U]$. Since $i \le d / 3$, at least $d / 2 - d / 3 = d / 6$ of these neighbors are not visited by the path $P_{i-1}$. Denote by $E(u_{i-1})$ this set of at least $d / 6$ edges connecting $u_{i-1}$ to not yet visited neighbors in $G[U]$. We now claim that at least one of the edges of $E(u_{i-1})$ is contained in $F$ with high probability. Observe that at stage $i-1$, we have only learnt so far that the sample $F$ contains the $i-1$ edges of the path $P_{i-1}$. Hence, 
\begin{align*}
\Pr & \left[ F  \cap E(u_{i-1}) = \varnothing \ | \ P_{i-1} \subseteq F \right]  = \frac{{ m - (i-1) - |E(u_{i-1})| \choose |F| - (i-1)}}{ { m - (i-1) \choose |F| - (i-1)}}  \\
& \le \exp \left(-  \frac{|E(u_{i-1})| (|F| - (i-1))}{m-(i-1)} \right) \le \exp \left(-  \frac{\frac{d}{6} (10 n \ln(n) - n \ln(n))}{m} \right) \\
& = \exp(- \frac{\frac{m}{3n} (9 n \log n)}{m}) = \frac{1}{n^{3}} \ ,
\end{align*}
where we used the inequality $\frac{{a-c \choose b}}{{a \choose b}} \le \exp \left(-\frac{bc}{a} \right)$ (see Lemma~\ref{lem:tech-1} in the appendix) to obtain the first inequality and the bound $(i-1) \le n \ln(n)$ to obtain the second.

We can therefore extend the path with probability $1-\frac{1}{n^3}$ at any step $i$. Since we run $d / 3 \le n$ steps overall to create the final path $P_{d / 3}$ of length $d / 3$, by the union bound, we succeed with probability at least $1 - \frac{1}{n^2}$. 

Last, since a longest path in $G$ is of length at most $n$, the algorithm also constitutes an $O(n/d)$-approximation algorithm to \textsf{Longest Path}.
\end{proof}

\subsection{Implementation in Streaming Models}\label{sec:implementations}
Algorithm~\ref{alg:sampling} can easily be implemented in both the insertion-only and the insertion-deletion streaming models. In the insertion-only model, a uniform sample of the edges in the stream can be obtained using {\em reservoir sampling} \cite{v85}, and, in the insertion-deletion model, this can be achieved using $\ell_0$-samplers \cite{jst11} and rejection sampling. 

We obtain the following theorem:

\newcounter{thmsaved}
\setcounter{thmsaved}{\value{theorem}}
\setcounter{theorem}{\value{counterALG}}
\addtocounter{theorem}{-1}
\begin{theorem}
 In both the insertion-only and the insertion-deletion models, for undirected graphs, there are semi-streaming algorithms that compute a path of length at least $d/3$ with high probability, where $d$ is the average degree of the input graph.
\end{theorem}
\setcounter{theorem}{\value{thmsaved}} 

Given a space bound $s = \Omega(n \poly \log n)$, we can run the previous algorithm in parallel with the trivial algorithm that stores $s$ edges. Then, if the input graph has at most $s$ edges then we obtain an exact solution (computed in exponential time as a post-processing step), and otherwise, we obtain an $O(n / d)$-approximation. Observe that, if the input graph has more than $s$ edges then its average degree $d$ is at least $d \ge 2s/n$, which implies that we find an $O(n^2 / s)$-approximation. Phrased differently, in order to obtain an $\alpha$-approximation algorithm, space $\tilde{O}(n^2 / \alpha)$ suffices. We thus obtain the following corollary:

\setcounter{thmsaved}{\value{corollary}}
\setcounter{corollary}{\value{counterCOR}}
\addtocounter{corollary}{-1}

\begin{corollary}
 In both the insertion-only and the insertion-deletion models, for undirected graphs, there are $\tilde{O}(n^2 / \alpha)$-space streaming algorithms that compute an $\alpha$-approximation to \textsf{Longest Path}.
\end{corollary}
\setcounter{corollary}{\value{thmsaved}} 

\section{Insertion-only Lower Bound for Directed Graphs} \label{sec:lb-directed}
In this section, we  prove an $\Omega(n^2)$ space lower bound for one-pass streaming algorithms for $\textsf{LP}$ on directed graphs that compute an $(n^{1 - o(1)})$-approximation.

To this end, we work with two input distributions $\mathcal{D}_{\text{SLP}}$ (\textbf{S}imple \textsf{\textbf{L}ongest \textbf{P}ath}) and $\mathcal{D}_{\text{LP}}$ (\textsf{\textbf{L}ongest \textbf{P}ath}). In Subsection~\ref{sec:simple}, we  give a lower bound on the information cost of protocols that solves $\mathcal{D}_{\text{SLP}}$ well. We achieve this by, first, proving a lower bound on the communication cost of any protocol that solves $\mathcal{D}_{\text{SLP}}$ directly via combinatorial arguments, and then employ a message compression argument that allows us to conclude that the information cost of such protocols  must also be large. Then, in Subsection~\ref{sec:complex}, we present the distribution $\mathcal{D}_{\text{LP}}$, which makes use of an $(r, t)$-RS-graph, and we establish a direct sum argument, showing that the information cost of protocols that solve $\mathcal{D}_{\text{LP}}$ is at least $t$ times the information cost of protocols that solve $\mathcal{D}_{\text{SLP}}$, which bounds the information cost of protocols that solve  $\mathcal{D}_{\text{LP}}$ from below. Then, since information cost is a lower bound on communication cost, we obtain our result.

\subsection{A Simple Distribution} \label{sec:simple}
We will first work with the distribution denoted $\mathcal{D}_{\text{SLP}}(r)$, see Figure~\ref{fig:dist-slp}.

\begin{figure}[h]
\begin{center} 
\fbox{
\begin{minipage}{0.96\textwidth}
\textbf{Input Distribution $\mathcal{D}_{\text{SLP}}(r)$:}

\vspace{0.1cm}

 The directed graph $G=(A, B = B_1 \cup B_2, E)$ with $|A| = |B_1| = |B_2| = r$, $A = \{a_1, \dots, a_r\}$, $B_1 = \{b^1_1, \dots, b^1_r\}$, $B_2 = \{b^2_1, \dots, b^2_r\}$, and $E= E_A \cup E_B$, where $E_A$ are Alice's edges and $E_B$ are Bob's edges, is obtained as follows:
 \vspace{0.2cm}
 
 \textbf{Alice's Input:} Edge set $E_A$
 
 \vspace{0.1cm}
 
 For each $i \in [r]$, flip an unbiased coin $X_i \in \{0, 1\}$ and if it comes out heads then insert the directed edge $(a_i, b^1_i)$ into the graph. If it comes out tail then insert the edge $(a_i, b^2_i)$. Observe that this constitutes a directed matching $M$ that matches all $A$-vertices and exactly one of $b_i^1, b_i^2$, for all $i$.

 \vspace{0.1cm}

Alice holds the edges $E_A = M$.

   \vspace{0.2cm}
 
 \textbf{Bob's Input:} Edge set $E_B$
 
 Let $N^1$ be a uniform random matching between $A$ and $B_1$, directed from $B_1$ towards $A$. Let $N^2$ be a copy of $N^1$, but every $B_1$-vertex is replaced by the corresponding $B_2$-vertex.
 
 \vspace{0.1cm}
 
 Bob holds the edges $E_B = N^1 \cup N^2$.
 
\end{minipage}
} \caption{Input Distribution $\mathcal{D}_{\text{SLP}}(r)$ \label{fig:dist-slp}} \end{center}
\end{figure}

We first argue that the expected length of a longest path in $H \sim \mathcal{D}_{\text{SLP}}(r)$ is $\Omega(r)$. 
\begin{lemma}\label{lem:golomb}
 The expected length of a longest path in $H \sim \mathcal{D}_{\text{SLP}}(r)$ is bounded as follows:
 $$\mathop{\Exp}_{H \gets \mathcal{D}_{\text{SLP}}(r)} \lp(H) \ge 2 \lambda r \ge 1.24 r \ , $$ 
 where $\lambda = 0.62432 \dots$ is the Golomb-Dickman constant.
\end{lemma}
\begin{proof}
For an input graph $H=(A, B_1 \cup B_2, E) \sim \mathcal{D}_{\text{SLP}}(r)$ with $B_1 = \{b_1^1, \dots, b_r^1 \}$ and $B_2 = \{b_1^2, \dots, b_r^2 \}$, let $H'$ be the graph obtained from $H$ by {\em contracting} all the vertex pairs $b^1_i$ and $b^2_i$, 
for every $i$, and by treating parallel edges in the resulting graph as single edges, see Figure~\ref{fig:slp} for an illustration. It is then easy to see that  $\lp(H') = \lp(H)$. 
Furthermore, denote by $\mathcal{D}'_{\text{SLP}}(r)$ the distribution of $H'$. We observe that both $\mathcal{D}_{\text{SLP}}(r)$ and $\mathcal{D}'_{\text{SLP}}(r)$ are uniform distributions.

 Next, consider the set $\Sigma_r$ of permutations of the set $\{1, 2, \dots, r\}$. Consider now the bijection $f: \Sigma_r \rightarrow \text{range}(\mathcal{D}'_{\text{SLP}}(r))$, where a permutation $\sigma \in \Sigma_r$ is mapped to the graph that contains the edges $(b_i, a_{\sigma(i)})$, for all $i$.
 Then, we observe that, for a permutation $\sigma \in \Sigma_r$, the length of the longest cycle $\lc(\sigma)$ is related to $\lp(f(\sigma))$ as follows: 
 $$2 \cdot \lc(\sigma) - 1 = \lp(f(\sigma)) \ . $$
 
 It is known that the length of a longest cycle in a random permutation $\sigma \in \Sigma_{r}$ is at least $\lambda \cdot r$, where $\lambda = 0.624 \dots$ is the Golomb-Dickman constant \cite{gwg59,g64}. Hence, we obtain
 $$\mathop{\Exp}_{\mathcal{H}' \sim \mathcal{D}'_{\text{SLP}}(r)}  \lp(H') \ge 2 \cdot \lambda \cdot r - 1 \ge 1.24 \cdot r \ ,$$ using the assumption that $r$ is large enough. Since the longest paths in $H$ and $H'$ are identical, the result follows.
\end{proof}

\begin{figure}[h]
\begin{center} \begin{minipage}{0.53\textwidth}
 \hspace{0.4cm} \includegraphics[height=4.5cm]{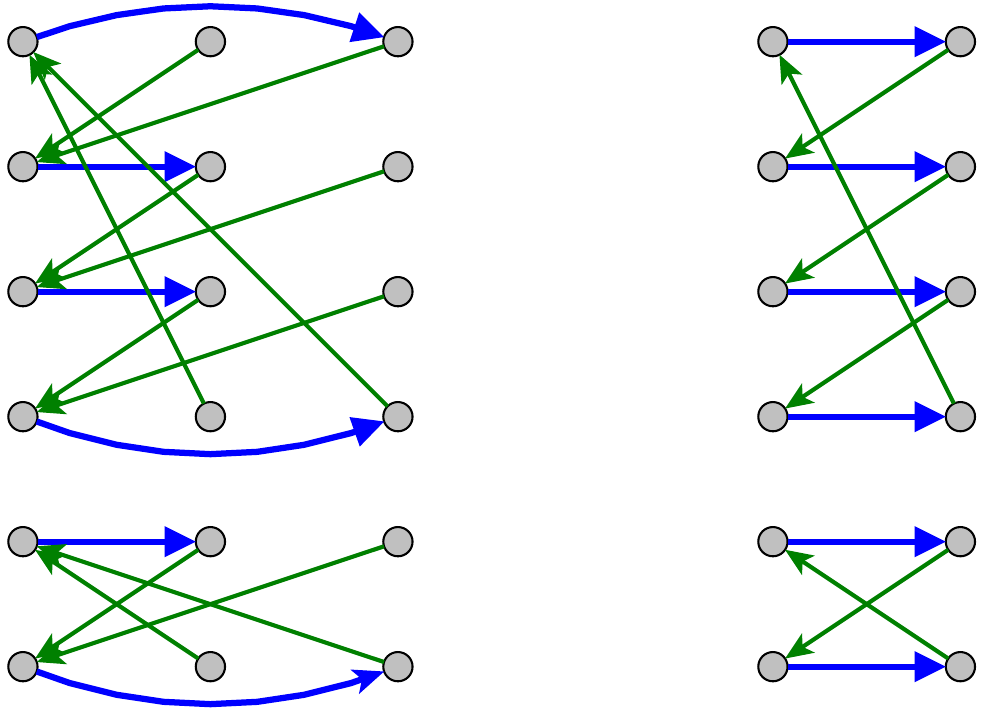}
 
 \vspace{-4.7cm} $a_1$ \hspace{2.7cm} $b_1^2$ \hspace{0.9cm} $a_1$ \hspace{1.45cm} $b_1$
 
 \vspace{0.33cm}
 $a_2$ \hspace{2.7cm} $b_2^2$ \hspace{0.9cm} $a_2$ \hspace{1.45cm} $b_2$
 
  \vspace{0.33cm}
 $a_3$ \hspace{2.7cm} $b_3^2$ \hspace{0.9cm} $a_3$ \hspace{1.45cm} $b_3$

\vspace{0.33cm}
 $a_4$ \hspace{2.7cm} $b_4^2$ \hspace{0.9cm} $a_4$ \hspace{1.45cm} $b_4$

\vspace{0.33cm}
 $a_5$ \hspace{2.7cm} $b_5^2$ \hspace{0.9cm} $a_5$ \hspace{1.45cm} $b_5$

 \vspace{0.33cm}
 $a_6$ \hspace{2.7cm} $b_6^2$ \hspace{0.9cm} $a_6$ \hspace{1.45cm} $b_6$
 
 \vspace{-5.1cm}
 
 \hspace{0.4cm} $A$ \hspace{0.55cm} $B_1$ \hspace{0.55cm} $B_2$ \hspace{1.7cm} $A$ \hspace{0.65cm} $B$
 
 \vspace{0.27cm}
 
 \hspace{1.9cm} $b_1^1$

 \vspace{0.25cm}
 
 \hspace{1.9cm} $b_2^1$

 \vspace{1.1cm}
 
 \hspace{1.9cm} $b_4^1$

 \vspace{0.28cm}
 
 \hspace{1.9cm} $b_5^1$

 \vspace{0.3cm}
 
 \hspace{1.9cm} $b_6^1$
 
 \vspace{0.3cm}
 
 \hspace{1.7cm} $H$ \hspace{3.3cm} $H'$
 \end{minipage}\begin{minipage}{0.2\textwidth}
 $$\sigma = \left( \begin{tabular}{cccccc} 
           1 & 2 & 3 & 4 & 5 & 6 \\           
           2 & 3 & 4 & 1 & 6 & 5
          \end{tabular}
\right) $$
 \end{minipage}
 \end{center}
 \caption{Illustration of the proof of Lemma~\ref{lem:golomb}. The vertices $b_i$ are obtained by contracting $b_i^1$ and $b_i^2$. We observe that the permutation $\sigma$ has a longest cycle of length $4$ ($1-2-3-4$), while $H'$ (and $H$) have longest paths of lengths $7$ ($a_1, b_1, a_2, b_2, \dots, a_4, b_4$). \label{fig:slp}} 
\end{figure}

Next, we show that any deterministic protocol on input distribution $\mathcal{D}_{\text{SLP}}(r)$ that communicates at most $\frac{1}{100} r$ bits outputs a path of length at most $O(\log r)$ with large probability over $\mathcal{D}_{\text{SLP}}(r)$.
\begin{lemma} \label{lem:det-alg}
 Let $\Pi_{\text{SLP}}$ be a deterministic one-way communication protocol for $\textsf{Longest Path}$ on distribution $\mathcal{D}_{\text{SLP}}(r)$ that communicates at most $ \frac{r}{100}$ bits. Then, the probability over the input distribution $\mathcal{D}_{\text{SLP}}(r)$ that $\Pi_{\text{SLP}}$  outputs a path of length $O(\log r)$ is at least $1 - \frac{1}{500}$.
\end{lemma}
\begin{proof}
 Denote by $\mathcal{M}$ the set of possible directed input matchings $M$ for Alice in $\mathcal{D}_{\text{SLP}}(r)$. Then, $|\mathcal{M}| = 2^r$. Next, since every message sent by protocol $\Pi_{\text{SLP}}$ is of length at most $s := r / 100$, there are at most $2^s$ different messages. On average, a message therefore corresponds to $2^{r-s}$ inputs, and, at most $\frac{1}{1000} \cdot 2^r$ inputs yield messages that in turn only correspond to at most $\frac{1}{1000} 2^{r-s}$ inputs. Hence, at least $\frac{999}{1000} 2^r$ inputs yield messages that each correspond to at least $\frac{1}{1000} 2^{r-s}$ inputs. Consider now one of these messages $\pi$, and let $M_1, M_2, \dots$ be the matchings $M$ that correspond to $\pi$.
 
 Given $\pi$, the protocol can only output an edge if it is contained in all matchings $M_i$. Suppose that there are $k$ such edges. Then, there are at most $2^{r-k}$ input graphs that contain these edges. We then have:
 $$2^{r-k} \ge \frac{1}{1000} 2^{r-s} \ , $$
 which implies $k \le s + 10$. 
 
 Denote by $K$ this set of at most $k \le s+10$ edges. We will now prove that the longest path in $K \cup E_B$ is of length $O(\log r)$ with high probability. 

Denote by $A'$ the $A$-endpoints of the edges $K$, and let $a_0' \in A'$ be any vertex. We  bound the length of the path  with starting point $a_0'$. This path first uses the edge in $K$ incident on $a_0'$, and denote by $b_0'$ the other endpoint of this edge. Then, the path uses the edge of $N_1$ or $N_2$ incident on $b_0'$ to return to an $A$-vertex that we denote by $a_1'$. Observe that we can only continue on this path if $a_1' \in A'$. If this is the case then we can continue in the same fashion, visit a $B$-vertex $b_1'$, and return to another $A$-vertex $a_2'$. Again, we can only continue if $a_2' \in A'$, and so on. 

For any $i \ge 1$, the following holds:

 %Denote by $A'$ the $A$-endpoints of the edges $K$. Let $a_0 \in A'$ be any vertex, denote by $\mathcal{P}$ the longest path starting at $a_0$, and let $a_0, a_1, a_2, \dots, a_p$ denote the $A'$-vertices visited by $\mathcal{P}$ in this order. Then, it can then be seen that $\mathcal{P}$ is of length at most $2 \cdot (p+2)$. Then, for any $i$:
$$\Pr [a_i' \in A' \ | \ a_0', \dots, a_{i-1}' \in A' ] \le \frac{s+10 - i}{r-i} \le \frac{s+10}{r} = \frac{1}{100} + \frac{10}{r} \le \frac{1}{50} \ , $$ 
 assuming that $r$ is large enough.
 Hence, the probability that the path contains $p+2$ $A'$-vertices and is thus of length $2(p+1)$ is at most $\frac{1}{50^p}$. This further implies that, with probability at least $1-\frac{1}{r^2}$, the path is of length $O(\log r)$. 
 
 Observe that the argument above applies when considering any start vertex in $A'$. Hence, by a union bound over all possibly start vertices in $A'$, all paths starting at an $A'$ vertex are of length at most $O(\log r)$ with probability at least $1- \frac{1}{r}$. Last, observe that the longest path that may be able to form could also start at a $B$-vertex. Such a path however is only by one edge longer than a path starting at an $A$-vertex. The $O(\log r)$ length bound therefore still holds.
 
 So far, we have proved that, for a message that corresponds to at least $\frac{1}{1000} 2^{r-s}$ inputs, Bob can only output a path of length at most $O(\log r)$ with high probability, where the probability is taken over Bob's input. Hence, overall, for a uniform input from $\mathcal{D}_{\text{SLP}}(r)$, the probability that a path of length $O(\log r)$ will be outputted is at least $\frac{999}{1000} \cdot (1-\frac{1}{r}) \ge \frac{998}{1000}$, where we assumed that $r$ is large enough.
 \end{proof}
We use the notation $\textsf{LP}^{\alpha}_{\epsilon}$ to denote  \textsf{LP}  with approximation factor $\alpha$ and error probability $\epsilon$.

\begin{theorem}\label{thm:cc-slp}
 The randomized one-way communication complexity of $\textsf{LP}^{O(r / \log r)}_{1/4}$ on inputs from $\mathcal{D}_{\text{SLP}}(r)$ is $\Omega(r)$. 
\end{theorem}
\begin{proof}
 Towards a contradiction, let $\Pi_{\text{r}}$ be a randomized protocol for \textsf{Longest Path} on inputs from $\mathcal{D}_{\text{SLP}}(r)$ with approximation factor $o(r / \log r)$ that errs with probability at most $1/4$. Given $\Pi_{\text{r}}$, by Yao's Lemma, there exists a deterministic protocol $\Pi_{\text{d}}$ on distribution $\mathcal{D}_{\text{SLP}}(r)$ with distributional error $1/4$ and approximation factor $o(r /  \log r)$, i.e., on at least $3/4$ of the inputs, the protocol achieves a $o(r / \log r)$-approximation.
 
Lemma~\ref{lem:golomb} states that $\Exp_{H \gets \mathcal{D}_{\text{SLP}}(r)} \lp(H) \ge 1.24 r$. This allows us to bound the quantity $\Pr_{H \gets \mathcal{D}_{\text{SLP}}(r)}[\lp(H) \ge \frac{1}{2}r]$, as follows:
\begin{align*}
 1.24 r \le \mathop{\Exp}_{H \gets \mathcal{D}_{\text{SLP}}(r)} \lp(H)  & \le \Pr_{H \gets \mathcal{D}_{\text{SLP}}(r)}[\lp(H) \ge \frac{1}{2}r] \cdot 2r + (1 - \Pr_{H \gets \mathcal{D}_{\text{SLP}}(r)}[\lp(H) \ge \frac{1}{2}r])  \frac{1}{2} r \\
 & = r \cdot \left(1.5 \cdot \Pr_{H \gets \mathcal{D}_{\text{SLP}}(r)}[\lp(H) \ge \frac{1}{2}r] + \frac{1}{2}  \right) \ , 
\end{align*}
where we used the fact that the longest path in $H$ is at most the number of vertices in $H$, i.e., $2r$. The previous inequality then implies that $\Pr_{H \gets \mathcal{D}_{\text{SLP}}(r)}[\lp(H) \ge \frac{1}{2}r] \ge \frac{1.24 - 0.5}{1.5} \ge \frac{1}{3}$.

Next, Lemma~\ref{lem:det-alg} states that, with probability at least $1-\frac{1}{500}$, $\Pi_{\text{d}}$ outputs a path of length at most $O(\log r)$. Hence, with probability at least $1/3 - 1/500 > 1/4$, there simultaneously exists a longest path of length at least $\frac{1}{2} r$ and $\Pi_{\text{d}}$ outputs one of length $O(\log r)$. The approximation factor of $\Pi_{\text{d}}$ is therefore $\Omega(r / \log r)$, contradicting the fact that $\Pi_{d}$ achieves an approximation factor of $o(r / \log r)$ on at least $3/4$ of the instances. Protocols $\Pi_d$ and $\Pi_r$ therefore do not exist, which completes the proof.
\end{proof}

\begin{lemma}\label{lem:simple-comm}
Let $\Pi$ be a randomized protocol for $\textsf{LP}^{O(r / \log r)}_{1/4-\frac{1}{100}}$ on inputs from $\mathcal{D}_{\text{SLP}}(r)$. Then:
 $$\ICost_{\mathcal{D}_{\text{SLP}}(r)}(\Pi) = \Omega(r) \ . $$
\end{lemma}
\begin{proof}
 Denote by $\Pi$ a randomized one-way two-party communication protocol for $\textsf{LP}^{O(r / \log r)}_{1/4-\frac{1}{100}}$ on inputs from $\mathcal{D}_{\text{SLP}}(r)$. Given $\Pi$, using the message compression technique stated in Theorem~\ref{thm:compression}, we obtain a protocol $\Pi'$ for $\textsf{LP}^{O(r / \log r)}_{1/4-\frac{1}{100}}$ that sends a message of expected size (recall that $M$ is Alice's input matching)
\begin{align}
s := I_{\mathcal{D}_{\text{SLP}}(r)}(M \ : \ \Pi) + 2 \cdot \log \left(1 + I_{\mathcal{D}_{\text{SLP}}(r)}(M \ : \ \Pi)  \right) + O(1) \ ,  \label{eqn:392}
\end{align}
 where the expectation is taken over the randomness used by the protocol. Then, by the Markov inequality, the probability that the message is of size at least $100 \cdot s$ is at most $\frac{1}{100}$. 

Given $\Pi'$, we now construct a protocol $\Pi''$ with maximum message size $100 \cdot s$ that solves \textsf{LP} with slightly increased error: Whenever $\Pi'$ sends a message of size at most $100 \cdot s$, then $\Pi''$ also sends this message and the protocol behaves exactly the same as $\Pi'$. However, when $\Pi'$ would send a message of size at least $100 \cdot s$ then $\Pi''$ aborts and thus fails. Since the probability of sending a message of size larger than $100 \cdot s$ is $1/100$, the error probability of $\Pi''$ is $(\frac{1}{4} - \frac{1}{100}) + \frac{1}{100} = \frac{1}{4}$.

From Theorem~\ref{thm:cc-slp}, we obtain that the communication cost of $\Pi''$ is $\Omega(r)$, which implies that $100 \cdot s = \Omega(r)$. Hence, combined with Inequality~\ref{eqn:392}, we conclude that
\begin{align*}
  \Omega(r) = I_{\mathcal{D}_{\text{SLP}}(r)}(M \ : \ \Pi) \le  I_{\mathcal{D}_{\text{SLP}}(r)}(M \ : \ \Pi \ | \ R) = \ICost_{\mathcal{D}_{\text{SLP}}(r)}(\Pi) \ ,
\end{align*}
where the inequality follows from property \textbf{P1} stated in the preliminaries.

\end{proof}
\subsection{Hard Input Distribution and Direct Sum Argument} \label{sec:complex}
Our hard input distribution $\mathcal{D}_{\text{LP}}(n)$ that is used to obtain our main lower bound result is stated in Figure~\ref{fig:dist-lp}.

%We prove our lower bound for directed graphs in the insertion-only model in the one-way two-party communication setting. In this setting, Alice holds a set of edges $E_A$, Bob holds a set of edges $E_B$, Alice sends a single message to Bob, who, upon receipt, outputs a long path in the graph spanned by the edges $E_A \cup E_B$.

%Our input distribution for Alice and Bob is constructed as follows:

\begin{figure}[h]
\begin{center} 
\fbox{
\begin{minipage}{0.96\textwidth}
\textbf{Input Distribution $\mathcal{D}_{\text{LP}}(n)$:}

\vspace{0.2cm}

 Let $H=(A^H, B^H, E^H)$ be an $(r,t)$-RS graphs with $|A^H| = |B^H| = n$ and induced matchings $M^H_1, M^H_2, \dots, M^H_t$.
 
 \vspace{0.3cm}
 
 \textbf{Alice's Input:} Edge set $E_A$
 
 \vspace{0.1cm}

   Given $H=(A^H, B^H, E^H)$, we construct the directed bipartite graph $G = (A, B = B_1 \ \dot{\cup} \ B_2, E)$, where:
   
%   \vspace{0.1cm}
  \begin{itemize}
   \item $A$ is a copy of $A^H$, and $B_1$ and $B_2$ are copies of $B^H$.
   \item For every edge $e = \{a,b\} \in E^H$, we flip an unbiased coin $X_{a,b}$. If it comes out heads then the directed edge $(a,b)$ between the sets $A$ and $B_1$ is introduced, and if it comes out tail then the directed edge $(a,b)$ between the sets $A$ and $B_2$ is introduced.
  \end{itemize}
  
  \vspace{0.05cm}
  We denote the edges in $E$ that originated from the matching $M^H_i$, for any $i$, by $M_i$. Alice holds the edges $E_A = \cup_i M_i$.

   \vspace{0.3cm}
 
 \textbf{Bob's Input:} Edge set $E_B$
 \begin{itemize}
  \item Let $J \in [t]$ be a uniform random index.
  \item Consider the matching $M^H_J \in E_H$ and let $A_H' = A(M^H_J)$ and $B_H' = B(M^H_J)$. Let $N_J$ be a uniform random matching between $A'$ and $B'$. 
  \item Bob introduces two copies of $N_J$ into $G$: To this end, let $A' \subseteq A$ be the copy of $A'_H$ in $G$, let $B'_1 \subseteq B_1$ be the copy of $B'$ in $B_1$, and let $B'_2 \subseteq B_2$ be the copy of $B'$ in $B_2$. The first copy $N_J^1$ is introduced between $A'$ and $B'_1$, and the second copy $N_J^2$ is introduced between $A'$ and $B'_2$. The edges in $N_J^1$ and $N_J^2$ are directed such that the $A$-vertex constitutes the head and the $B$-vertex the tail of the directed edge.
 \end{itemize}  
 \vspace{0.05cm}
 Bob holds the edges $E_B = N_J^1 \cup N_J^2$.
\end{minipage}
} \caption{Input Distribution $\mathcal{D}_{\text{LP}}(n)$ \label{fig:dist-lp}} 
\end{center}
\end{figure}

We will now argue that a protocol $\Pi_{\text{LP}}(n)$ that solves $\text{LP}$ under distribution $\mathcal{D}_{\text{LP}}(n)$ can be used to obtain a protocol $\Pi_{\text{SLP}}(r)$ that solves $\text{LP}$ under distribution $\mathcal{D}_{\text{SLP}}(r)$, where $r$ is the size of the induced matchings of the RS-graph used in $\mathcal{D}_{\text{LP}}(n)$. This establishes a connection between the information costs of $\Pi_{\text{LP}}(n)$ and $\Pi_{\text{SLP}}(r)$. To this end, we use the reduction stated in Algorithm~\ref{alg:direct-sum}.

\begin{algorithm}[h]
 \begin{algorithmic}[1]
  \REQUIRE $ $ \\ \vspace{-0.2cm} \begin{enumerate}
            \item Protocol $\Pi_{\text{LP}}(n)$ that solves $\textsf{LP}^{\alpha}_{1/4}$ under distribution $\mathcal{D}_{\text{LP}}(n)$ \vspace{-0.2cm}
            \item Input $(M, N^1, N^2) \sim \mathcal{D}_{\text{SLP}}(r)$, Alice holds matching $M$, Bob holds matchings $N^1, N^2$
           \end{enumerate}
           \vspace{-0.2cm}
  \STATE Alice and Bob use public randomness to generate a uniform random index $J \in [t]$
  \STATE Alice samples graph $G \sim \mathcal{D}_{\text{LP}}(n)$ and updates matching $M_J$ in $G$ such that $M_J = M$
  \STATE Bob set $N_J^1 = N^1$ and $N_J^2 = N^2$ and adds these to $G$
  \STATE Alice and Bob run the protocol $\Pi_{\text{LP}}(n)$ on $G$, denote by $\mathcal{P}$ the output of $\Pi_{\text{LP}}(n)$
  \RETURN Longest sub-path of $\mathcal{P}$ that solely uses edges in $M_J \cup N_J^1 \cup N_J^2$
 \end{algorithmic}
 \caption{Construction of protocol $\Pi_{\text{SLP}}(r)$ \label{alg:direct-sum}}
\end{algorithm}

This reduction has the following properties:

\begin{lemma}\label{lem:struct}
 Consider the reduction given in Algorithm~\ref{alg:direct-sum}. Then, given any path $\mathcal{P}$ in $G$ of length at least $2$, the path $\mathcal{P}'$ with the first and last edge removed only uses edges from $M_J \cup N_J^1 \cup N_J^2$. %This implies the inequality 
 %$$\lp(G) \le \lp(H) + 2 \ . $$
\end{lemma}
\begin{proof}
 Let $\mathcal{P}$ be a path of length at least $2$ in $G$. Since $G$ is a bipartite graph, every other edge in $\mathcal{P}$ must be an edge from $N_J^1 \cup N_J^2$ since these edges are the only edges with head in $A$. Observe that the edges $M_J$ are the only edges with both endpoints in $A(N_J^1 \cup N_J^2)$ and $B(N_J^1 \cup N_J^2)$. Hence, for any three consecutive edges $e,f,g$ in path $\mathcal{P}$, if $e$ and $g$ are edges from $N_J^1 \cup N_J^2$, then $f$ must be an edge from $M_J$. 
 
 These conditions imply that if an edge that is not contained in $M_J \cup N_J^1 \cup N_J^2$ is included in $\mathcal{P}$ then this edge must necessarily be either the first or the last edge of $\mathcal{P}$.
\end{proof}

Given our reduction, we now relate the information costs of $\Pi_{\text{SLP}}(r)$ and $\Pi_{\text{LP}}(n)$: %distributions $\mathcal{D}_{\text{LP}}$ and $\mathcal{D}_{\text{SLP}}$:

\begin{lemma}\label{lem:direct-sum}
Let $\Pi_{\text{LP}}(n)$ be a protocol for $\textsf{LP}^\alpha_{\epsilon}$, for some parameters $\alpha$ and $\epsilon$, and let $\Pi_{\text{SLP}}(r)$ be the protocol obtained from $\Pi_{\text{LP}}(n)$ via the reduction given in Algorithm~\ref{alg:direct-sum}. Then, $\Pi_{\text{SLP}}(r)$ has approximation factor $O(\alpha)$, errs with the same probability $\epsilon$, and:
 \begin{align*}
\ICost_{\mathcal{D}_{\text{LP}}(n)}(\Pi_{\text{LP}}(n)) = t \cdot \ICost_{\mathcal{D}_{\text{SLP}}(r)}(\Pi_{\text{SLP}}(r)) \ .  
 \end{align*}
\end{lemma}
\begin{proof}
We denote by $R_{\text{SLP}}$ and $R_{\text{LP}}$ the public randomness used in $\Pi_{\text{SLP}}(r)$ and in $\Pi_{\text{LP}}(n)$, respectively. Then (see the rules \textbf{P1}, \dots, \textbf{P4} in the preliminaries):
 \begin{align*} \displaystyle
  \ICost_{\mathcal{D}_{\text{SLP}}(r)}(\Pi_{\text{SLP}}(r)) & = I_{\mathcal{D}_{\text{SLP}}(r)}(M \ : \ \Pi_{\text{SLP}}(r) \ | \ R_{\text{SLP}}) \\
  & = I_{\mathcal{D}_{\text{SLP}}(r)}(M_J \ : \ \Pi_{\text{LP}(n)} \ | \ R_{\text{LP}}, J) \\
  & = I_{\mathcal{D}_{\text{LP}}(n)}(M_J \ : \ \Pi_{\text{LP}(n)} \ | \ R_{\text{LP}}, J) \\
  & = \Exp_{j \gets J} I_{\mathcal{D}_{\text{LP}}(n)}(M_J \ : \ \Pi_{\text{LP}(n)} \ | \ R_{\text{LP}}, J=j) & \textbf{P2} \\
  & = \frac{1}{t} \cdot \sum_{j \in [t]} I_{\mathcal{D}_{\text{LP}}(n)}(M_j \ : \ \Pi_{\text{LP}(n)} \ | \ R_{\text{LP}}, J=j) \\
  & = \frac{1}{t} \cdot \sum_{j \in [t]} I_{\mathcal{D}_{\text{LP}}(n)}(M_j \ : \ \Pi_{\text{LP}(n)} \ | \ R_{\text{LP}}) & \textbf{P3}  \\
  & \le \frac{1}{t} \cdot \sum_{j \in [t]} I_{\mathcal{D}_{\text{LP}}(n)}(M_j \ : \ \Pi_{\text{LP}(n)} \ | \ M_1, \dots, M_{j-1}, R_{\text{LP}}) & \textbf{P1} \\  
  & = \frac{1}{t} \cdot I_{\mathcal{D}_{\text{LP}}(n)}(M_1, \dots, M_t \ : \ \Pi_{\text{LP}(n)} \ | R_{\text{LP}})  &  \textbf{P4}\\    
  & = \ICost_{\mathcal{D}_{\text{LP}}(n)}(\Pi_{\text{LP}}(n)) \ . 
 \end{align*}
 Regarding the approximation factor, observe that $\lp(H) \le \lp(G)$.  %by Lemma~\ref{lem:struct}, $\lp(G) \le \lp(H) +2$. 
 Furthermore, by Lemma~\ref{lem:struct}, the path found by $\Pi_{\text{SLP}}(r)$ is at most by $2$ shorter than the path found by $\Pi_{\text{LP}}(n)$, which in turn is of length at least $\lp(G) / \alpha$. Hence, the approximation factor of $\Pi_{\text{SLP}}(r)$ is bounded by:
 \begin{align*}
  \frac{\lp(H)}{\lp(G) / \alpha -2} \le \frac{\lp(G)}{\lp(G) / \alpha -2} = O(\alpha) \ .
 \end{align*}
\end{proof}
We are now ready to state our main lower bound result for directed graphs.
\begin{theorem}
 The randomized one-way communication complexity of $\textsf{LP}^{O(r / \log r)}_{1/4-\frac{1}{500}}(n)$ is $\Omega(r \cdot t)$. 
\end{theorem}
\begin{proof}
Let $\Pi_{\text{LP}}(n)$ be a protocol for $\textsf{LP}^{O(r / \log r)}_{1/4-\frac{1}{500}}(n)$. Then, from Lemmas~\ref{lem:direct-sum} and \ref{lem:simple-comm}, we obtain 
$$\ICost_{\mathcal{D}_{\text{LP}}(n)}(\Pi_{\text{LP}}(n)) = \Omega(r \cdot t) \ .$$
%Using the RS-graphs by Alon et al. \cite{ams12} as stated in Theorem~\ref{thm:rs-graphs}, we obtain that $r \cdot t = \Omega(n^2)$. 
Since the choice of $\Pi_{\text{LP}}(n)$ was arbitrary, and information complexity is a lower bound on communication complexity, the result follows.
\end{proof}

Last using the well-known connection between streaming algorithms and one-way two-party communication protocols as well as the RS-graph construction by Alon et al. as stated in Theorem~\ref{thm:rs-graphs}, we obtain the following theorem:

\setcounter{thmsaved}{\value{theorem}}
\setcounter{theorem}{\value{counterLB-directed}}
\addtocounter{theorem}{-1}

\begin{theorem}
 Every one-pass streaming algorithms for $\textsf{LP}^{O(n^{1-o(1)} / \log(n^{1-o(1)})}_{1/4-\frac{1}{500}}(n) = \textsf{LP}^{O(n^{1-o(1)} )}_{1/4-\frac{1}{500}}(n)$ requires space $\Omega(n^2)$. 
\end{theorem}

\setcounter{theorem}{\value{thmsaved}}

%\textit{Remark:} One may wonder whether an $\Omega(n^2)$ space lower bound can also be proved for approximation factors above $O(n / \log n)$. This, however, is not possible and our lower bound is tight. To see this, consider the algorithm that samples every edge with constant probability $p > 0$ such that the space used by this sample is below the space indicated by our $\Omega(n^2)$ lower bound. 

%In the following lemma, we establish key properties of (long) paths in $G$:
%\begin{lemma}
% Let $G \sim \mathcal{D}_{LP}$, and let $\mathcal{P}$ be a path in $G$. Then:
% \begin{enumerate}
%  \item The edges of $\mathcal{P}$ alternate between $E_A$ edges  and $E_B$ edges.
%  \item $\mathcal{P}$ contains at most two edges that are not contained in $M_J \cup N_J^1 \cup N_J^2$.
%  \item A graph chosen from $\mathcal{D}_{LP}$ contains a long path with high (constant) probability, i.e., for any $\delta > 0$, the following holds: $$\Pr_{G \sim \mathcal{D}_{LP}}[G \text{ contains a path of length } 2 \cdot \delta \cdot s] \ge 1 - \delta \ . $$
% \end{enumerate}
%\end{lemma}

\section{Insertion-only Lower Bound for Undirected Graphs} \label{sec:lb-undirected}
We will prove our lower bound for undirected graphs via a reduction to the \textsf{Index}$(N)$ problem in the one-way two-party communication model, see Definition~\ref{def:index} and Theorem~\ref{thm:lb-index} for its hardness.
We will show that a one-way two-party protocol for \textsf{LP}$(n)$ on undirected graphs with small approximation factor can be used to solve \textsf{Index}$(N)$. This allows us to obtain a lower bound on the message size of any protocol that solves \textsf{LP}$(n)$, which in turn establishes a space lower bound for one-pass streaming algorithms.

Our reduction is presented in Figure~\ref{fig:reduction}.

\begin{figure}
 \fbox{\begin{minipage}{0.98\textwidth}
  \textbf{Input:} \vspace{-0.2cm} 
  \begin{itemize}
   \item Protocol $\Pi_{\text{LP}}$ that solves $\textsf{LP}^{\alpha}_{\epsilon}(n)$,
        with approximation ratio $\alpha$ and error probability $\epsilon$
   \vspace{-0.2cm}
   \item  Input $X \in \{0, 1\}^N, J \in [N]$ for $\textsf{Index}(N)$ such that Alice holds $X$ and Bob holds $J$ \vspace{-0.2cm}
   \item Let $H = (A^H, B^H, E^H)$ be an $(r,t)$-RS graph with $|A^H| = |B^H| = n$ with induced matchings $M_1^H, M_2^H, \dots, M_t^H$ each of size $r$ such that $N = r \cdot t$ and denote by $e_{ij}$ the $j$th edge in $M_i^H$. 
  \end{itemize}

  \textbf{Alice and Bob:} Alice and Bob use public randomness to compute a uniform random permutation
  $\pi:[r] \rightarrow [r]$ and a uniform random binary vector $Y \in \{0, 1\}^{N}$. 
  
  \vspace{0.3cm}
  
  \textbf{Alice:} Given input $X \in \{0, 1\}^N$, random bits $Y$, and random permutation $\pi$, Alice constructs the following graph $G' = (A \cup B_1 \cup B_2, E)$ with $|A| = |B_1| = |B_2| = n$:
  
  \begin{enumerate}
   \item  For each $i \in [t], j \in [r]$, if $(X \oplus Y)_{(i-1) \cdot r + \pi(j)} = 0$ then introduce the edge $e_{ij}$ into $G'$ between $A$ and $B_1$, otherwise introduce the edge $e_{ij}$ between $A$ and $B_2$. Denote by $M_i$ the set of edges that originated from matching $M_i^H$.
   \item Alice runs the protocol $\Pi_{\text{LP}}$ on $G'$ and sends the respective message to Bob.
  \end{enumerate}
  
\textbf{Bob:} Given input $J \in [N]$, Bob extends the construction as follows. Let $i^* \in [t],j^* \in [r]$ be such that $(i^*-1) \cdot r + j^* = J$. 

Furthermore, let $A^{H'} = A(M^H_{i^*})$ and $B^{H'} = B(M^H_{i^*})$, i.e., $A^{H'}$ (respectively, $B^{H'}$)  is the subset of $A$-vertices (respectively $B$-vertices) 
of $H$ that are incident to the matching $M^H_{i^*}$.
Let $A' \subseteq A$ be the copy of $A^{H'}$ in $G'$, let $B'_1 \subseteq B_1$ be the copy of $B^{H'}$ in $B_1$, and let $B'_2 \subseteq B_2$ be the copy of $B^{H'}$ in $B_2$. 

\begin{enumerate}
 \item Let $\mathcal{P}'$ be a path consisting of new edges (not in the edge set $E$ of Alice's graph) that visits all the vertices in $(A \setminus A') \cup (B_1 \setminus B_1') \cup (B_2 \setminus B_2')$ in this order. %, i.e., first $A \setminus A'$, then $B_1 \setminus B_1'$, then $B_2 \setminus B_2'$. 
  Then, let $\mathcal{P}$ be the {\em subdivided} path $\mathcal{P}'$, where every edge in $\mathcal{P}'$ is replaced by a path of length $\ell$, for some constant integer $\ell$ whose value we will determine later. Note that these subdivisions introduce $\ell -1$ new vertices for every edge in $\mathcal P'$.
  Let $s(\mathcal{P})$ and $t(\mathcal{P})$ be the start and end vertices in $\mathcal{P}$. Let $F$ be the set of edges connecting $t(\mathcal{P})$ to every vertex in $A' \cup B_1' \cup B_2'$. 
 
 \item Consider the matching $M^H_{i^*} \in E^H$. Let $N_{i^*}$ be a uniform random matching between $A^{H'}$ and $B^{H'}$. 
  Bob obtains two copies of $N_{i^*}$ and adds them to $G'$ as follows: The first copy $N_{i^*}^1$ is introduced between $A'$ and $B'_1$, and the second copy $N_{i^*}^2$ is introduced between $A'$ and $B'_2$. 
  
 \item Bob continues the execution of $\Pi_{\text{LP}}$ by adding the edges $\mathcal{P} \cup F \cup N_{i^*}^1 \cup N_{i^*}^2$ to the input. This establishes the final input graph $G= G' \cup \mathcal{P} \cup F \cup N_{i^*}^1 \cup N_{i^*}^2$. 
 
 \item Denote by $\mathcal{Q}$ the path outputted by $\Pi_{\text{LP}}$. If $\mathcal{Q}$ contains the edge that corresponds to $e_{{i^*} \pi^{-1}({j^*})}$ then Bob can successfully output the result: Let $Z$ be the indicator random variable of the event ``$e_{i^* \pi^{-1}(j^*)}$ is incident on $B_2$ in $G$''. Then, output $Z \oplus Y_{(i^*-1)\cdot r + \pi^{-1}(j^*)}$. Otherwise, output a uniform random bit.
 
\end{enumerate}

 \end{minipage}}
\caption{Reduction between \textsf{Longest Path} and \textsf{Index}. \label{fig:reduction}}
\end{figure}

We will first argue that the graph $G$ constructed in our reduction (Figure~\ref{fig:reduction}) contains a long path with probability at least $1/3$, where the probability is over the distribution of $N_{i^*}^1$ and $N_{i^*}^2$:
\begin{lemma}
With probability at least $1/3$ (over the distribution of $N_{i^*}^1$ and $N_{i^*}^2$):
 $$\lp(G) \ge (3 \cdot (n-r) - 1) \cdot \ell + 1 + \frac{1}{2} r  \ . $$
\end{lemma}
\begin{proof}
 A path of this length can be constructed as follows:
 
 First, we follow the path $\mathcal{P}$, which is of length $(3 \cdot (n-r) - 1) \cdot \ell$. This is because the path $\mathcal{P}'$ visits $3 \cdot (n-r)$ vertices  and thus consists of $3 \cdot (n-r) - 1$ edges. Each of these edges is then replaced by a path of length $\ell$ to give the path $\mathcal{P}$.
 
 Then, as argued in the proof of Theorem~\ref{thm:cc-slp}, we know that with probability at least $1/3$ (over $N_{i^*}^1$ and $N_{i^*}^2$), there exists a path $\mathcal{R}$ of length at least $0.5 r$ among the edges $M_{i^*} \cup N_{i^*}^1 \cup N_{i^*}^2$. We will now follow the edge in $F$ that connects the endpoint of the path $\mathcal{P}$, i.e.,  $t(\mathcal{P})$, to the path $\mathcal{R}$, which adds $1 + \frac{1}{2} r$ edges to the path and yields the result. 
\end{proof}
Next, we will argue that if the path $\mathcal{Q}$ outputted by the protocol is long then it must contain many edges of $M_{i^*}$.

\begin{lemma}
 Consider the output $\mathcal{Q}$ of the protocol. Then, for every $\ell \ge 4$:
 $$|\mathcal{Q}| \le 3 \cdot (n-r) \cdot \ell + 2 \cdot |\mathcal{Q} \cap M_{i^*}| + 2\ell \ . $$
\end{lemma}
\begin{proof}
Let $V' = A' \cup B_1' \cup B_2'$. Our argument is based on the following observation: Consider a vertex $u \in (V \setminus V') \setminus \{s(\mathcal{P}), t(\mathcal{P}) \}$  that is connected to $v \in V'$ in the path $\mathcal{Q}$. Since $u$ is neither the start nor the end vertex of the path $\mathcal{P}$, it is incident to two {\em hoops}, i.e., subdivisions on the path $\mathcal{P}$. However, since $u$ connects to $v$ in $\mathcal{Q}$, $\mathcal{Q}$ cannot visit one of these subdivisions, and, hence, the $\ell-1$ vertices within this subdivision are not visited by $\mathcal{Q}$, which limits the overall length of $\mathcal{Q}$. On the other hand, if $\mathcal{Q}$ contains only few edges across the cut $(V, V')$ and many vertices of $V'$ are visited by $\mathcal{Q}$ then $\mathcal{Q}$ must necessarily contain many $M_{i^*}$ edges. We now quantify these observations.

 %Let $k = |\mathcal{Q} \cap M_j|$, and let $V' = A' \cup B_1' \cup B_2'$. Furthermore, let $x$ be the number of edges of $\mathcal{Q}$ that go across the cut $(V \setminus V', V')$. 

 Let $\{u_1, v_1\}, \{u_2, v_2\}, \dots, \{u_f, v_f\}$ with $u_i \in V \setminus V'$ and $v_i \in V'$ be all the edges of $\mathcal{Q}$ that go across the cut $(V \setminus V', V')$. 
 
 We first observe that if $u_i$ is different to $s(\mathcal{P})$ and $t(\mathcal{P})$ then $u_i$ is incident to two hoops or subdivisions. Furthermore, since $\{u_i, v_i \}$ is an edge on $\mathcal{Q}$, one of the two hoops is not visited by $\mathcal{Q}$. Some care needs to be taken since it is possible that the path $\mathcal{Q}$ starts and ends with such a hoop. Hence, we have to account for overall two hoops that may be additionally visited. Next, observe that a hoop is incident on two vertices. Hence, overall, we have established that the $\ell-1$ vertices of at least $f/2 - 2$ hoops are not visited by $\mathcal{Q}$.
 
 Next, let $k = |\mathcal{Q} \cap M_{i^*}|$. We will now argue that at most $2k + f/2 + 1$ edges of $M_{i^*} \cup N_{i^*}^1 \cup N_{i^*}^2$ are contained in $\mathcal{Q}$. To this end, consider a {\em maximal} subpath $x_1, x_2, \dots, x_t$ of $\mathcal{Q}$ such that all vertices $x_1, \dots, x_t$ are contained in $V'$, and suppose that there are overall $q$ such paths. Suppose that the subpath $x_1, \dots, x_t$ uses $k'$ edges of $M_{i^*}$. Then, the subpath $x_1, \dots, x_t$ uses at most $2k' + 1$ edges of $M_{i^*} \cup N_{i^*}^1 \cup N_{i^*}^2$. Hence, all these $q$ maximal subpaths consist of at most $2k + q$ edges. Next, in order to connect two of these subpaths in $\mathcal{Q}$, at least two cut edges are needed, which implies that there are at most $q \le \frac{f}{2}+1$ such subpaths. Hence, these maximal subpaths are of combined length at most $2k + \frac{f}{2} + 1$. % we observe that, since there are $f$ cut edges, and each path $u_1, x_1, \dots, x_t, u_2$ starts and ends with such an edge, there are overall at most $\frac{f}{2}$ such paths. Hence, summing up over all these $\frac{f}{2}$ paths, we see that at most  $2k + \frac{f}{2}$ edges of $M_{i^*} \cup N_{i^*}^1 \cup N_{i^*}^2$ are contained in $\mathcal{Q}$.

 The total length of $\mathcal{Q}$ is thus bounded as follows: 
 \begin{align*}
|\mathcal{Q}| & \le |V \setminus V'| + f + 2k + \frac{f}{2} + 1 -  (f/2 - 2) \cdot (\ell-1) \\  
& \le 3 (n - r) \cdot \ell + 2k + f(2  - \ell / 2) +2 \ell - 1 \\
& \le 3 (n - r) \cdot \ell + 2k + 2 \ell  \ ,
 \end{align*}
for every $\ell \ge 4$, which completes the proof.

\end{proof}

\setcounter{thmsaved}{\value{theorem}}
\setcounter{theorem}{\value{counterLBundirected}}
\addtocounter{theorem}{-1}

\begin{theorem}
 Let $\Pi_{\text{LP}}$ be a one-way two-party communication protocol for $\text{LP}^{\frac{25}{24} - \gamma}_{\epsilon}(n)$, for any $\gamma > 0$ and $0 < \epsilon < 1$. Then, $\Pi_{\text{LP}}$ has communication cost $n^{1 + \Omega(\frac{1}{\log \log n})}$.
\end{theorem}
\setcounter{theorem}{\value{thmsaved}}
\begin{proof}
 We employ the RS-graph construction of Goel et al. with parameters $r = (\frac{1}{2} - \epsilon)n$ and $t = n^{\Omega(\frac{1}{\log \log n})}$, for some small constant $\epsilon > 0$. We also set $\ell = 4$ in our construction. Then, with probability at least $1/3$, the approximation factor of the protocol is at least:
 \begin{align*}
  \frac{25}{24} - \gamma \ge & \frac{\lp(G) }{|\mathcal{Q}|} \ge \frac{\left( 3(n-(\frac{1}{2} - \epsilon) n)-1 \right) \cdot 4 + 1 + \frac{1}{2} \cdot (\frac{1}{2} - \epsilon) n}{(3 \cdot (n-(\frac{1}{2} - \epsilon) n)) \cdot 4 + 2 k + 2 \cdot 4} \\
  & = \frac{n( 6 + \frac{1}{4} + 11.5 \epsilon) - 3 }{n(6 + 12 \epsilon) + 2k + 8}  \ ,
 \end{align*}
which implies that $k = \Omega(n)$, for small enough $\epsilon > 0$. 

Next, observe that, due to the random permutation $\pi$ and the XOR operation with the random uniform bits $Y$, every edge of $M_{i^*} \cap \mathcal{Q}$ is equally likely to have originated from the bit $X[J]$. Hence, the probability that one of the $k$ output edges of $M_{i^*}$ corresponds to the bit $X[J]$ is $k / r = \Omega(1) =: p$. Hence, with potentially small but constant probability $p$, the algorithm successfully reports the output bit. Otherwise, the algorithm outputs a uniform random bit. The overall probability of learning the bit at position $J$ is thus at least:

$$\frac{1}{3} \cdot p \cdot (1 - \epsilon) + (1-\frac{1}{3} \cdot p \cdot (1 - \epsilon)) \cdot \frac{1}{2} =  \frac{1}{2} + \frac{1}{3} p \cdot (1-\epsilon) = \frac{1}{2} + \Omega(1) \ . $$

By Theorem~\ref{thm:lb-index}, the protocol thus is required to use a message of size $\Omega(N)$. Since $N = \Theta(r \cdot t) = n^{1+\Omega(\frac{1}{\log \log n})}$, and the graph $G$ has $\Theta(\ell \cdot n) = \Theta(n)$ vertices the result follows.
\end{proof}

\section{Insertion-deletion Lower Bound for Undirected Graphs} \label{sec:lb-insert-delete} 
Our lower bound for Insertion-deletion streams for undirected graphs is a reduction to the \textsf{Augmented-Index} problem in the one-way two-party communication setting, see Definition~\ref{def:index} and Theorem~\ref{thm:lb-index} for its hardness. %Similar to \textsf{Index}$(N)$, in \textsf{Augmented-Index}$(N)$ Alice holds a bit-string $X \in \{0, 1\}^N$ of length $N$. Bob, however, holds both an index $J \in [N]$ as well as the suffix $X[J+1, N]$. The goal for Bob is to output the bit $X[J]$.

%Similar to \textsf{Index}$(N)$, it is known that every protocol that solves \textsf{Augmented-Index}$(N)$ with probability bounded away from $1/2$ requires a message of size $\Omega(N)$.

%\begin{theorem}\label{thm:aug-index}
% Every randomized protocol for \textsf{Augmented-Index}$(N)$ that errs with probability at most $1/2 - \delta$, for any constant $\delta > 0$, requires a message of size $\Omega(N)$.
%\end{theorem}

%We will now give a reduction from \textsf{Augmented-Index} to \textsf{Longest Path}. 
Our reduction is such that Bob introduces edge deletions into the resulting \textsf{Longest Path} instance. This reduction therefore only implies a lower bound for insertion-deletion streaming algorithms.
See  Figure~\ref{fig:reduction-insert-delete} for our reduction.

\begin{figure}
 \fbox{\begin{minipage}{0.98\textwidth}
  \textbf{Input:} \vspace{-0.2cm} 
  \begin{itemize}
   \item Protocol $\Pi_{\text{LP}}$ that solves $\textsf{LP}^{\alpha}_{\epsilon}(n)$
    with approximation factor $\alpha$ and error probability $\epsilon$
   and can tolerate edge deletions \vspace{-0.2cm}
   \item  Input $X \in \{0, 1\}^N, J \in [N]$ for $\textsf{Augmented-Index}(N)$ such that Alice holds $X$ and Bob holds $J$ and $X[J+1, N]$
  \end{itemize}
%We write $X_{i,j}$ to denote the element $X[(i-1) \cdot \sqrt{N} + j]$.

\vspace{0.3cm}

  \textbf{Alice and Bob:} Alice and Bob use public randomness to compute two uniform random permutations
  $\pi_1, \pi_2:[n] \rightarrow [n]$ and a uniform random binary matrix $Z \in \{0, 1\}^{n \times n}$.  Let $Y' \in \{0, 1\}^{n \times n}$ be a binary matrix such that the bits $X$ of the index instance are embedded in the top-left $\sqrt{N} \times \sqrt{N}$ submatrix (ordered from left to right and then top to bottom). All other entries are uniform random bits chosen from public randomness that both Alice and Bob know. Then, let $Y = Y' \oplus Z$.

  \vspace{0.3cm}
  
  \textbf{Alice:} Given input $Y$ and random permutations $\pi_1, \pi_2$, Alice constructs the following graph $G' = (A \cup B_1 \cup B_2, E)$ with $|A| = |B_1| = |B_2| = n$ and $A = \{a^1, \dots, a^{n}\}$, $B_1 = \{b_1^1, \dots, b_1^n \}$, and $B_2 = \{b_2^1, \dots, b_2^n\}$:
  
  \begin{enumerate}
   \item  For each $i,j \in [n]$, if $Y_{i, j} = 0$ then introduce the edge $\{a^{\pi_1(i)},b_1^{\pi_2(j)} \}$ into $G'$, otherwise introduce the edge $\{a^{\pi_1(i)}, b_2^{\pi_2(j)} \}$ into $G'$. 
   
   \item Alice runs the protocol $\Pi_{\text{LP}}$ on $G'$ and sends the respective message to Bob.
  \end{enumerate}
  
\textbf{Bob:} Given input $J \in [N]$, $\pi_1, \pi_2$ and the portion of $Y$ known to Bob, Bob extends the construction as follows. Let $i^*, j^* \in [\sqrt{N}]$ be such that $(i^*-1) \cdot \sqrt{N} + j^* = J$.

Observe that, for every $i,j$ with $(i-1) \sqrt{N} + j > J$, Bob knows whether Alice inserted the edge $\{a^{\pi_1(i)}, b_1^{\pi_2(j)}\}$ or $\{a^{\pi_1(i)}, b_2^{\pi_2(j)} \}$ as this choice depends on the bit $Y_{i,j}$ which Bob knows.
\begin{enumerate}
 \item \textbf{Edge Deletions.} Consider the following set of indices:
 \begin{align*}
 \mathcal{I}  =  & \{(i, j) \ : \ i \ge i^*, j \ge j^* \} \\ & \setminus \biggl( \{(i^* + k,j^* + k) \ : \  k \ge 0\} \cup \{(i^* + 1 + k,j^* + k) \ : \  k \ge 0\} \biggr) \ .
 \end{align*}
 
 For every $(i,j) \in \mathcal{I}$, Bob deletes the edge $\{a^{\pi_1(i)}, b_1^{\pi_2(j)} \}$ if $Y_{i,j} = 0$, and the edge $\{a^{\pi_1(i)}, b_2^{\pi_2(j)} \}$ if $Y_{i,j} = 1$. Denote this set of edge deletions by $F$.
 
 \item \textbf{Edge Insertions.} For every $(i,j)$ such that $i = i^* + 1 + k$ and $j = j^* + k$, for some integer $k \ge 0$, Bob introduces the edge $\{a^{\pi_1(i)}, b_1^{\pi_2(j)} \}$ if $Y_{i,j} = 1$, and the edge $\{a^{\pi_1(i)}, b_2^{\pi_2(j)} \}$ if $Y_{i,j} = 0$. Denote this set of edges $E_2$.
  
 \item Bob continues the execution of $\Pi_{\text{LP}}$ by adding the edges $E_2$ and by deleting the edges $F$. This establishes the final input graph $G= G' \cup E_2 \setminus F$.  
 
 \item Denote by $\mathcal{Q}$ the path outputted by $\Pi_{\text{LP}}$. If $\mathcal{Q}$ contains the edge
$\{a^{\pi_1(i^*)}, b_1^{\pi_2(j^*)} \}$ then output $Z_{i^*,j^*}$ and if $\mathcal{Q}$ contains the edge $\{a^{\pi_1(i^*)}, b_2^{\pi_2(j^*)} \}$ then output $1 - Z_{i^*,j^*}$. Otherwise, output \texttt{fail}.
\end{enumerate}

 \end{minipage}}
\caption{Reduction between \textsf{Longest Path} and \textsf{Index} for insertion-deletion streams. \label{fig:reduction-insert-delete}}
\end{figure}

In our analysis, we make use of the following sets of edges: 
\begin{align*}
M  & := \biggl\{ \{a^{\pi_1(i)}, b_{Y_{i,j} +1 }^{\pi_2(j)} \} \ : \ i = i^* + k, j = j^* + k, j \ge 0   \biggr\}  \\
 N_1  & := \biggl\{ \{a^{\pi_1(i)}, b_1^{\pi_2(j)} \} \ : \ i = i^* + 1 + k, j = j^* + k, j \ge 0   \biggr\} \\
 N_2  & := \biggl\{ \{a^{\pi_1(i)}, b_2^{\pi_2(j)} \} \ : \ i = i^* + 1 + k, j = j^* +k, j \ge 0   \biggr\} \ .
 \end{align*} 
 We observe that these are the only edges in the vertex-induced subgraph $G[A(M) \cup B(N_1) \cup B(N_2)]$.

We first argue that every graph $G$ created in our reduction contains a longest path of length $\Omega(n)$.
\begin{lemma}
 Every instance $G$ created in the reduction of Figure~\ref{fig:reduction-insert-delete} is such that $$\lp(G) \ge 2 \cdot(n - \sqrt{N}) - 1 \ .$$
\end{lemma}
\begin{proof}
 We observe that the set of edges $M \cup N_1 \cup N_2$ where $M$ forms a matching and the edges $N_1 \cup N_2$ connect the edges of $M$ contains a path of length $2 \cdot(n - \sqrt{N}) - 1$.
\end{proof}

Next, we show that if the path outputted by the protocol is long then it must contain many edges of $M$.
\begin{lemma}
Consider the output path $\mathcal{Q}$ produced by the protocol. Then:
 $$|\mathcal{Q}| \le 6 \sqrt{N} + 4 + 2 |M \cap \mathcal{Q}| \ . $$
\end{lemma}
\begin{proof}
 Let $A' \subseteq A(M)$ be the set of $A$-vertices visited in $\mathcal{Q}$ such that, for every $a \in A'$, the edge of $M$ with endpoint $a$ is {\em not} contained in $\mathcal{Q}$. Denote $B = B_1 \cup B_2$.
 Then, if $a \in A'$ is neither the start nor the end vertex of $\mathcal{Q}$ then $a$ is incident on an edge in $\mathcal{Q}$ with endpoint $b \in B-B(N_1 \cup N_2)$. Next, observe that $|B-B(N_1 \cup N_2)| \le 2 \cdot \sqrt{N}$. Hence, we obtain $|A'| \le 2 \cdot \sqrt{N} + 2$, where the additive $2$ accounts for potential start and end vertices in $A'$. 
 
 The path $\mathcal{Q}$ thus visits at most $2 \cdot \sqrt{N} + 2 + |M \cap \mathcal{Q}|$, $A(M)$-vertices and thus at most $2 \cdot \sqrt{N} + 2 + |M \cap \mathcal{Q}| + \sqrt{N}$, $A$-vertices. Last, since every other vertex visited in $\mathcal{Q}$ is a $B$-vertex, we obtain that the path $\mathcal{Q}$ is of length at most:
 
 $$2 \cdot \left( 2 \cdot \sqrt{N} + 2 + |M \cap \mathcal{Q}| + \sqrt{N} \right) = 6 \sqrt{N} + 4 + 2 |M \cap \mathcal{Q}| \ .$$
\end{proof}

Next, we show that the reduction is successful, i.e., it does not report \texttt{fail}, with probability at least $\frac{1}{4\alpha}$.

\begin{lemma}\label{lem:reductionProb}
 Let $N \le \beta \cdot \frac{n^2}{\alpha^2}$, where $\alpha$ is the approximation factor of protocol $\Pi_{\text{LP}}$ used in our reduction, and $\beta > 0$ a small enough constant. Then, the reduction succeeds with probability at least $\frac{1}{4\alpha}$.
\end{lemma}
\begin{proof}
 First, we assume that $\Pi_{\text{LP}}$ does not fail. Then, since the approximation factor of $\Pi_{\text{LP}}$ is $\alpha$, we obtain:
 \begin{align*}
  \alpha \ge \frac{2 \cdot(n - \sqrt{N}) - 1}{6 \sqrt{N} + 4 + 2 |M \cap \mathcal{Q}|} \,
 \end{align*}
which implies that

\begin{align*}
 |M \cap \mathcal{Q}| \ge \frac{1}{\alpha} \left( n - \sqrt{N} - \frac{1}{2} - 2\alpha - 3\alpha \sqrt{N} \right) \ge  n / 2\alpha \ ,
\end{align*}

using the assumption $N \le \beta \cdot \frac{n^2}{\alpha^2}$, for some small $\beta > 0$.

Next, we observe that, due to the random permutations and the bitwise XOR with a uniform random bitstring, the protocol cannot identify which of the edges of $M$ corresponds to the special index $(i^*, j^*)$, or, equivalently, to $J$. In other words, each of the edges of $M \cap \mathcal{Q}$ has the same probability to correspond to the index $J$. Since $|M \cap \mathcal{Q}| \ge  n / 2\alpha$ and $|M| \le n$, with probability at least $1 / 2\alpha$, we learn the bit $X[J]$. Last, since $\Pi_{\text{LP}}$ itself has an error probability of $\epsilon < \frac{1}{2}$, the result follows.
\end{proof}

\begin{theorem}
 Every one-way two-party communication protocol with approximation factor $\alpha$ for $\textsf{LP}^{\alpha}_{\epsilon}$ with deletions requires a message of size $\Omega(n^2 / \alpha^3)$.
\end{theorem}
\begin{proof}
 Consider a protocol as in the statement of the theorem. We will run the protocol $4\alpha$ times in parallel. It follows from Lemma~\ref{lem:reductionProb} 
 that the probability of not learning the output bit $X[J]$ is at most:
 $$(1 - \frac{1}{4\alpha})^{4\alpha} \le \exp \left( - \frac{1}{4\alpha} \cdot 4\alpha \right) \le 1 / e \le 0.3679 \ . $$

Hence, with probability at least $1-0.3679 \ge 0.6$ we solve the \textsf{Augmented-Index} instance. 

Since the \textsf{Augmented-Index} instance is of size $N = \Theta(\frac{n^2}{\alpha^2})$, by Theorem~\ref{thm:lb-index}, the $4\alpha$ parallel runs of our reduction must use an overall message of size $\Omega(\frac{n^2}{\alpha^2})$. This in turn implies that at least one of the messages used by $\Pi_{\text{LP}}$ in our reduction is of size $\Omega(\frac{n^2}{\alpha^3})$, which completes the proof.
\end{proof}

We are now ready to give our lower bound for insertion-deletion streams.
\setcounter{thmsaved}{\value{theorem}}
\setcounter{theorem}{\value{counterLB-deletion}}
\addtocounter{theorem}{-1}

\begin{theorem}
 Every one-pass insertion-deletion streaming algorithm for \textsf{LP} on undirected graphs with approximation factor $\alpha \ge 1$ requires space $\Omega(n^2 / \alpha^3)$. 
\end{theorem}
\setcounter{theorem}{\value{thmsaved}}

\section{Conclusion} \label{sec:conclusion}
In this paper, we studied one-pass streaming algorithms for the \textsf{Longest Path} problem. We showed that, in both insertion-only and  insertion-deletion streams, for undirected graphs, there are semi-streaming algorithms that find paths of lengths at least $\frac{d}{3}$ with high probability, where $d$ is the average degree of the input graph. The algorithm can also give an $\alpha$-approximation algorithm that uses space $\tilde{O}(n^2 / \alpha)$. 
We then showed that no such result can be obtained for directed graphs in that a $n^{1-o(1)}$-approximation requires space $\Omega(n^2)$, even in insertion-only streams. We also showed that semi-streaming algorithms in the insertion-only model for undirected graphs cannot yield an arbitrarily small constant factor approximation, and we showed that, in insertion-deletion streams, space $\Omega(n^2 / \alpha^3)$ is necessary to obtain an $\alpha$-approximation in undirected graphs. We conclude with two open questions: 

First, while we resolved the space complexity for one-pass streaming algorithms for directed graphs in both the insertion-only and the insertion-deletion models, the space complexity for undirected graphs, in particular, in the insertion-only model, remains wide open. While the offline setting is orthogonal to the streaming setting, one cannot help but wonder whether similar mechanisms are at work that prevent us from obtaining stronger NP-hardness results for \textsf{Longest Path} approximation in undirected graphs and from obtaining stronger space lower bounds in the streaming setting.
Can we either prove stronger lower bounds or give more space efficient algorithms for undirected graphs in the insertion-only setting? 

Second, for undirected graphs, are there multi-pass semi-streaming algorithms that allow us to compute longer paths than $\Theta(d)$, where $d$ is the average degree of the input graph?

\bibliography{kt25}

\appendix

\section{Technical Lemmas}

\begin{lemma}\label{lem:tech-1}
Let $a,b,c$ be positive with $a-c \ge b$. Then:
 $$\frac{{a-c \choose b}}{{a \choose b}} \le \exp(- \frac{b \cdot c}{a}) \ . $$
\end{lemma}
\begin{proof}
 We compute:
 \begin{align*}
  \frac{{a-c \choose b}}{{a \choose b}} & = \frac{\frac{(a-c)!}{(a-b-c)!b!}}{\frac{a!}{(a-b)! b!}} = \frac{(a-c) \cdot (a-c-1) \cdot \ldots \cdot (a-b-c+1)}{a \cdot (a-1) \cdot \ldots \cdot (a-b+1)} \\
  & \le \left( \frac{a-c}{a} \right)^b = (1- \frac{c}{a})^b \le \exp(-\frac{bc}{a}) \ ,
 \end{align*}
where we used the inequality $1+x \le \exp(x)$, which holds for all $x$.
\end{proof}

\begin{lemma} \label{lem:avg-min}
Let $G = (V, E)$ be a graph with $|V| = n$, $|E| = m$, and average degree $d = 2 \frac{m}{n}$. Then, there exists a subset of vertices $U \subseteq V$ such that the vertex-induced subgraph $G[U]$ has minimum degree greater than $\frac{d}{2}$.
\end{lemma}
\begin{proof}
 We iteratively remove vertices of degree at most $d/2$ from $G$ until no such vertex is left. Let $G_i$ be the graph with the first $i$ vertices removed, and let $G_0 = G$. We denote $m_i$ the number of edges in $G_i$, and $n_i = n - i$ the number of vertices in $G_i$. It can then be seen by induction that the average degree $d_i$ of every graph $G_i$ is still at least $d_0 = d$. Indeed, observe that removing a vertex of degree at most $d/2$ removes at most $d/2$ edges from the graph. Then, the average degree of graph $G_{i+1}$ is:
 \begin{align*}
     2 \cdot \frac{m_{i+1}}{n_{i+1}} \ge 2 \cdot \frac{m_i - \frac{d}{2}}{n - (i+1)} \ge 2 \cdot \frac{\frac{d}{2} (n - i) - \frac{d}{2}}{n-(i+1)} = d \ .
 \end{align*}  
 Then, since the average degree remains as high as $d$ throughout, and we only ever remove vertices of degree at most $d/2$, the process must leave a non-empty graph with minimum degree at least $\lfloor d/2 + 1 \rfloor$ behind. The set $U$ then is the set of vertices that are not removed by this process.
\end{proof}

\end{document}